%% file: ms.tex
\documentclass[11pt]{article}

\usepackage[letterpaper,margin=1.0in]{geometry}
\usepackage{bm, amsmath, amsthm, amssymb, amsfonts}

\usepackage{cite}
\usepackage{framed}
\usepackage[framemet hod=tikz]{mdframed}
\usepackage{appendix}
\usepackage{graphicx}
\usepackage{color}
\usepackage{varwidth}
\usepackage{wrapfig}
\usepackage{enumitem}

\usepackage[textsize=tiny]{todonotes}

\newcommand{\fabian}[1]{{\noindent\color{purple}[[Fabian: #1]]}}

\setitemize{noitemsep,topsep=3pt,parsep=3pt,partopsep=3pt}

\usepackage{xspace}
\usepackage{xifthen}

\definecolor{darkgreen}{rgb}{0,0.5,0}
\definecolor{darkblue}{rgb}{0,0,0.8}
\usepackage{hyperref}
\hypersetup{
    unicode=false,          
    colorlinks=true,        
    linkcolor=darkblue,          
    citecolor=darkgreen,        
    filecolor=magenta,      
    urlcolor=brown           
}
\RequirePackage[]{silence}
\usepackage{algorithm}
\usepackage[noend]{algpseudocode}

\usepackage[nameinlink,capitalize]{cleveref}

\newtheorem{theorem}{Theorem}[section]
\newtheorem{proposition}[theorem]{Proposition}
\newtheorem{lemma}[theorem]{Lemma}
\newtheorem{corollary}[theorem]{Corollary}
\newtheorem{definition}[theorem]{Definition}

\renewcommand{\vec}[1]{\ensuremath{\boldsymbol{#1}}}

\newcommand{\q}{\text{?}}

\newcommand{\calP}{\ensuremath{\mathcal{P}}}

\newcommand{\ignore}[1]{}

\algnewcommand\algorithmicswitch{\textbf{switch}}
\algnewcommand\algorithmiccase{\textbf{case}}

\algdef{SE}[SWITCH]{Switch}{EndSwitch}[1]{\algorithmicswitch\ #1\ \algorithmicdo}{\algorithmicend\ \algorithmicswitch}%
\algdef{SE}[CASE]{Case}{EndCase}[1]{\algorithmiccase\ #1}{\algorithmicend\ \algorithmiccase}%
\algtext*{EndSwitch}%
\algtext*{EndCase}%

\newcommand{\LOCAL}{\ensuremath{\mathsf{LOCAL}}\xspace}
\newcommand{\SLOCAL}{\ensuremath{\mathsf{SLOCAL}}\xspace}
\newcommand{\ZLOCAL}{\ensuremath{\mathsf{ZLOCAL}}\xspace}

\newcommand{\eps}{\varepsilon}

\newcommand{\Prob}[1]{\mathbb{P}\left(#1\right)}
\renewcommand{\Pr}{\mathbb{P}}
\newcommand{\set}[1]{\left\{#1\right\}}

\DeclareMathOperator{\E}{\mathbb{E}}
\DeclareMathOperator{\polylog}{polylog}
\DeclareMathOperator{\poly}{poly}

\newcommand{\complexityclass}[2][]{\ensuremath{\mathsf{#2}\ifthenelse{\isempty{#1}}{}{(#1)}}}

\newcommand{\loc}[1][]{\complexityclass[#1]{LOCAL}}
\newcommand{\zloc}[1][]{\complexityclass[#1]{ZLOCAL}}
\newcommand{\seqloc}[1][]{\complexityclass[#1]{SLOCAL}}

\newcommand{\randloc}[1][]{\complexityclass[#1]{RLOCAL}}
\newcommand{\randseqloc}[1][]{\complexityclass[#1]{RSLOCAL}}
\newcommand{\polyloc}{\complexityclass{P}\text{-}\complexityclass{LOCAL}}

\newcommand{\polyzloc}{\complexityclass{P}\text{-}\complexityclass{ZLOCAL}}
\newcommand{\polyseqloc}{\complexityclass{P}\text{-}\complexityclass{SLOCAL}}
\newcommand{\polyrandloc}{\complexityclass{P}\text{-}\complexityclass{RLOCAL}}

\newcommand{\hide}[1]{}

\newcommand{\FullOrShort}{full}

\ifthenelse{\equal{\FullOrShort}{full}}{
  
  \usepackage{times}
  \newcommand{\fullOnly}[1]{#1}
  \newcommand{\shortOnly}[1]{}

}{
  
  \usepackage{times}
  \newcommand{\shortOnly}[1]{#1}
  \newcommand{\fullOnly}[1]{}
}
  
\begin{document}

\title{On Derandomizing Local Distributed Algorithms}


\author{
 Mohsen Ghaffari\\
  \small ETH Zurich \\
  \small ghaffari@inf.ethz.ch
\and
David G. Harris \\
\small University of Maryland, College Park \\
\small davidgharris29@gmail.com
\and
Fabian Kuhn\footnote{Supported by ERC Grant No.\ 336495 (ACDC).}\\
  \small University of Freiburg \\
  \small kuhn@cs.uni-freiburg.de
 }

\date{}
\maketitle

\begin{abstract}
The gap between the known randomized and deterministic local
distributed algorithms underlies arguably the most fundamental and
central open question in \emph{distributed graph algorithms}. In this
paper, we combine the method of conditional expectation with
  network decompositions to obtain a generic and clean recipe for derandomizing randomized \LOCAL algorithms and transforming them into efficient deterministic \LOCAL algorithms. This simple recipe leads to significant improvements on a number of problems, in cases resolving known open problems. Two main results are:

\begin{itemize}
\item An improved deterministic distributed algorithm for hypergraph maximal matching, improving on Fischer, Ghaffari, and Kuhn~[FOCS'17], and giving improved algorithms for edge-coloring, maximum matching approximation, and low out-degree edge orientation. The last result gives the first positive resolution in the Open Problem 11.10 in the book of Barenboim and Elkin. 

\item Improved randomized and deterministic distributed algorithms for the Lov\'{a}sz Local Lemma, which gets closer to a conjecture of Chang and Pettie~[FOCS'17], and moreover leads to improved distributed algorithms for problems such as defective coloring and $k$-SAT.
\end{itemize}
\end{abstract}

\setcounter{page}{0}
\thispagestyle{empty}
\newpage

\section{Introduction and Related Work}
\label{sec:intro}

The gap between the best known deterministic and randomized local
distributed algorithms constitutes a central open question in the area
of \emph{distributed graph algorithms}. For many of the classic
problems (e.g., maximal independent set (MIS) and $(\Delta+1)$-vertex
coloring), $O(\log n)$-time randomized algorithms have been known
since the pioneering work of Luby\cite{luby86} and Alon, Babai, and
Itai\cite{alon86}. However, obtaining a $\polylog n$-time
deterministic algorithm remains an intriguing open problem for 30
years now, which was first asked by Linial\cite{linial1987LOCAL}. The
best known deterministic round complexity is $2^{O(\sqrt{\log n})}$,
due to Panconesi and Srinivasan\cite{panconesi95}.

The issue is not limited to a few problems; these are just symptomatic
of our lack of general tools and techniques for
derandomization. Indeed, in their 2013 book on Distributed Graph
Coloring\cite{barenboimelkin_book}, Barenboim and Elkin stated that
``\emph{perhaps the most fundamental open problem in this field is to
  understand the power and limitations of randomization}'', and left
the following as their first open problem\footnote{Their Open
  Problems 11.2 to 11.5 also deal with the same question, directly
  asking for finding efficient deterministic algorithms for certain
  problems; in all cases, the known randomized algorithms satisfy the
  goal.}:

\begin{center}
\begin{minipage}{0.95\linewidth}
\vspace{-8pt}
\begin{mdframed}[hidealllines=false, backgroundcolor=gray!00]
\textsc{Open Problem 11.1} (Barenboim \& Elkin\cite{barenboimelkin_book}) $\;$ Develop a general derandomization technique for the distributed message-passing model. 
\end{mdframed}
\end{minipage}
\end{center}

There is also a more modern and curious motive for developing
derandomization techniques for distributed algorithms, even if one
does not mind the use of randomness: efficient deterministic
algorithms can help us obtain even more efficient \emph{randomized}
algorithms. More concretely, most of the recent developments in
randomized algorithms use the \emph{shattering
  technique}\cite{barenboim2016locality,elkin20152delta,gmis,GS17,ghaffari-lll,harris2016coloring,chang2018optimal} which randomly breaks down the graph into
small components, typically $\polylog n$-size, and then solves them
separately via a deterministic algorithm. Once we develop faster
deterministic algorithms, say via derandomization, we can speed up the
corresponding randomized algorithms.

\smallskip

We next overview our contributions regarding this question. Before
that, we review Linial's \LOCAL model\cite{linial1987LOCAL,peleg00},
which is the standard synchronous message passing model of distributed
computing, and a sequential variant of it, called \SLOCAL, introduced
recently by Ghaffari, Kuhn, and Maus\cite{ghaffari2017complexity},
which will be instrumental in driving and explaining our results.

\paragraph{The \textsf{LOCAL} Model:} The communication network is
abstracted as an undirected $n$-node graph $G=(V, E)$ with maximum
degree $\Delta$, and where each node has a $\Theta(\log n)$-bit unique
identifier. Initially, nodes only know their neighbors in $G$. At the
end, each node should know its own part of the solution, e.g., its
color in coloring. Communication happens in synchronous rounds, where
in each round each node can perform some arbitrary internal
computations and it can exchange a possibly arbitrarily large message
with each of its neighbors. The \emph{time complexity} of an algorithm
is defined as the number of rounds that are required until all nodes
terminate. In the case of randomized algorithms, each node can in
addition produce an arbitrarily long private random bit string before
the computation starts.

A \LOCAL model algorithm with time complexity $r$ can alternatively be
defined as follows. When the computation starts, each node in
parallel reads the initial states of the nodes in the $r$-hop
neighborhood (including the private random bit strings for
randomized algorithms). Based on that information, each node in
parallel computes its output.

\paragraph{The \textsf{SLOCAL} Model:}
There are two main obstacles to developing \LOCAL algorithms: \emph{locality} and \emph{symmetry breaking.} In order to understand their roles separately,  Ghaffari, Kuhn, and Maus\cite{ghaffari2017complexity} introduced the \SLOCAL model,
in which symmetry breaking is free and only locality becomes a bottleneck.
The \SLOCAL is similar to the
\LOCAL model, in that each node can read its $r$-hop neighborhood in
the graph $G$, for some parameter $r$. However, in the \SLOCAL model,
the neighborhoods are read sequentially. Formally, the nodes are
processed in an arbitrary (adversarially chosen) order. When node $v$
is processed, $v$ can read its $r$-hop neighborhood and it computes
and locally stores its output $y_v$ and potentially additional
information. When reading the $r$-hop neighborhood, $v$ also reads all the
information that has been locally stored by the previously processed
nodes there. We call the parameter $r$ the
\emph{locality} of an \SLOCAL algorithm.

The \SLOCAL model can be seen as a natural extension of
sequential greedy algorithms. In fact, the classic distributed graph
problems such as MIS or $(\Delta+1)$-coloring have simple \SLOCAL
algorithms with locality $1$: in order to determine whether a node $v$
is in the MIS or which color $v$ gets in a $(\Delta+1)$-coloring, it
suffices to know the decisions of all the neighbors that have been
processed before.

The \SLOCAL model is inherently sequential. The main reason it is useful is that there
are transformations from \SLOCAL algorithms to \LOCAL algorithms,
which handles symmetry breaking in a ``black-box'' or generic way. By
developing and analyzing \SLOCAL algorithms, we are therefore able to treat a number of diverse \LOCAL problems in a unified and abstract fashion. We are
also able to adapt multiple types of algorithms to take advantage of special structure in the graphs (for example, bounds on its maximum degree).
This two-part method of algorithm analysis --- constructing \SLOCAL algorithms, and transforming them to \LOCAL algorithms generically --- will be
a key technical tool.

\subsection{Our Contributions, Part I: Derandomization}
In the first part of this paper, we present a simple and clean recipe for derandomizing local distributed algorithms. A simplified, and imprecise, version of our derandomization result is as follows: 

\begin{theorem}[Derandomization---Informal and
  Simplified] \label{thm:main_informal} Any $r$-round randomized
  \LOCAL algorithm for a locally checkable problem can be transformed
  to a deterministic \SLOCAL algorithm with locality $O(r)$. This
  \SLOCAL algorithm can then be transformed to a deterministic \LOCAL
  algorithm with complexity $\Delta^{O(r)} + O( r \log^* n),$
  or $\,r \cdot 2^{O(\sqrt{\log n})}$, by using network decompositions.\footnote{We provide the formal
  definition of network decompositions later in
  \Cref{sec:model}.}
\end{theorem} 

We show \Cref{thm:main_informal}  using the method of conditional
expectation. In the \SLOCAL model, the nodes are processed
sequentially; when we encounter vertex $v$, we fix its randomness (as used in the randomized distributed algorithm) to minimize the conditional expected number of failed nodes. It is critical for this that the problem solution is locally
checkable. We note that this assumption is generally necessary for
efficiently derandomizing distributed algorithms. In \Cref{prop:cyclemarking},  we provide a
simple problem that is not locally checkable,  which has a
constant-time randomized distributed algorithm, but where the
best deterministic \SLOCAL algorithm (or consequently \LOCAL algorithm) has at least polynomial
locality. 

The transformation from the \SLOCAL model to the \LOCAL model by using
network decompositions follows from \cite{ghaffari2017complexity}. It is
well-known that network decompositions are a powerful generic tool to
obtain efficient distributed algorithms.  \Cref{thm:main_informal}
proves that network decompositions are even sufficiently powerful to
turn any randomized distributed algorithm for a locally checkable
problem into an efficient deterministic algorithm. Indeed,  since $(O(\log n),O(\log n))$-decompositions
  can be computed in randomized polylogarithmic time in the \LOCAL
  model, \Cref{thm:main_informal} exactly captures the set of problems
  for which network decompositions can be used. 
  
   The
$2^{O(\sqrt{\log n})}$ term in the round complexity of
\Cref{thm:main_informal} is due to the currently best-known deterministic round
complexity of computing $(O(\log n), O(\log n))$-network
decompositions\cite{panconesi95}. Because of this overhead,
unfortunately, the deterministic algorithms obtained from this derandomization method may be much less efficient than their
randomized counterparts. However, as we overview next, still in many
cases we get algorithms that are far more efficient than existing
algorithms, and in a few cases, they resolve known open problems; this
is especially due to the first bound, combined with some other
ideas that effectively reduce the maximum degree $\Delta$.

\subsubsection{Hypergraph Maximal Matching and Ramifications}
Hypergraph maximal matching was recently pointed out by Fischer,
Ghaffari, and Kuhn\cite{FischerGK17} as a clean and powerful problem,
which admits reductions from several classic problems of the
area. As the first concrete implication of our derandomization
technique, we obtain an improved deterministic algorithm for maximal
matching in hypergraphs:

\begin{theorem}\label{thm:hypergraph} There is an
  $O(r^2 \log (n\Delta) \log n \log^4 \Delta)$-round deterministic
  \LOCAL algorithm that computes a maximal matching for any $n$-node
  hypergraph with maximum degree $\Delta$ and 
  rank at most $r$ (the rank is the maximum number of vertices in a
  hyperedge).
\end{theorem}

The algorithm of \cite{FischerGK17} has complexity
$(\log \Delta)^{O(\log r)} \cdot \log n$, which is efficiently only for hypergraphs of essentially constant rank.  Improving
this $r$-dependency was left as a main open question in
\cite{FischerGK17}. \Cref{thm:hypergraph} gives efficient algorithms for $r = \polylog n$.

Using known reductions\cite{FischerGK17}, this improved hypergraph maximal matching algorithm leads to efficient deterministic algorithms for a number of problems, including $(2\Delta-1)$-edge coloring, ($1-\eps$)-approximation of maximum matching, and low out-degree orientation: 

\begin{corollary}\label{crl:edge-coloring} There is a deterministic distributed algorithm that computes a $(2\Delta-1)$-list-edge-coloring in $O(\log^4 \Delta \log^2 n)$ rounds.
\end{corollary}

\begin{corollary}\label{crl:matching-approx} There is a deterministic distributed algorithm that computes a $(1-\eps)$-approximation of  maximum matching in $O(\log^2 n\log^5 \Delta/\eps^9)$ rounds.
\end{corollary}

\begin{corollary}\label{crl:orientation} There is a deterministic distributed algorithm that computes an orientation with maximum out-degree at most $(1+\eps)\lambda$ in any graph with arboricity $\lambda$, in $O(\log^{10} n\log^5 \Delta/\eps^9)$ rounds.
\end{corollary}

\Cref{crl:edge-coloring} gives an alternative solution for the problem
of $(2\Delta-1)$-edge coloring, which was a well-known open problem
since 1990s (listed as Open Problem 11.4 of
\cite{barenboimelkin_book}) and was resolved recently by Fischer et
al.\cite{FischerGK17}; indeed, the solution of
\Cref{crl:edge-coloring} is more efficient than the $O(\log^8 n)$
algorithm of \cite{FischerGK17}. Moreover, \Cref{crl:orientation}
gives the first positive resolution of Open Problem 11.10 of
\cite{barenboimelkin_book}. In a recent paper, Ghaffari,
Kuhn, Maus, and Uitto \cite{det_1+eps_edgecoloring} gave efficient
(polylogarithmic time) deterministic distributed algorithms for edge
coloring with $(1+\eps)\Delta$ and $\frac{3}{2}\Delta$ colors. We note
that these results critically depend on the stronger maximum matching
approximation, which is guaranteed by \Cref{crl:matching-approx}.

\subsubsection{The Lov\'{a}sz Local Lemma and Extensions}
\label{sec:LLLresults}
The Lov\'{a}sz Local Lemma (LLL) is a powerful probabilistic tool; at a
high level, it states that if one has a probability space $\Omega$ and
a collection $\mathcal B$ of ``bad'' events in that space, then as
long as the bad-events have low probability and are not too
interdependent, then there is a positive probability that no event in
$\mathcal B$ occurs; in particular a configuration avoiding
$\mathcal B$ exists.  In its simplest, ``symmetric'' form, it states
that if each bad-event has probability at most $p$, and each bad-event
affects at most $d$ other bad-events, and $e p d \leq 1$, then there is a
positive probability that no bad-event occurs.

A distributed algorithm for the LLL is a key building-block for a number of distributed graph coloring algorithms, such as frugal or defective vertex-colorings. In addition, as Chang \& Pettie noted in \cite{chang2017time}, the LLL plays a important role in the overall landscape of \LOCAL algorithms; the reason is that in any randomized $\loc[r]$ algorithm, one may define a bad-event that a given node $v$ fails; this bad-event has low probability, and because of the locality of the procedure it only affects other nodes within radius $2 r$. 

As a result, \cite{chang2017time} showed that if we have a distributed LLL algorithm running in time $t(n)$, then any \LOCAL algorithm on a bounded-degree graph for a locally checkable problem running in $o(\log_{\Delta} n)$ rounds, can be sped up to time $t(n)$. Thus, in a sense, the LLL is a universal sublogarithmic \LOCAL problem. They further conjectured that the LLL could be solved in time $O(\log \log n)$ (matching a lower bound of \cite{brandt}), which in turn would allow a vast range of other \LOCAL algorithms to run in $O(\log \log n)$ time.

There has been a long history of algorithmic versions of the LLL, including distributed algorithms. A breakthrough result of Moser \& Tardos \cite{mt} gave one of the first general sequential algorithms for it; they also discussed a parallel variant, which can easily be converted into a randomized $\loc[O(\log^2 n)]$ algorithm. There have been a number of other algorithms developed specifically in the context of the \LOCAL model \cite{pettie, ghaffari-lll}. These typically require satisfying a stronger condition than the LLL, of the form $p (2 d)^c \leq 1$, for some constant $c \geq 1$; we refer to as a \emph{polynomially-weakened LLL criterion} (pLLL). Most LLL constructions can be adapted to a pLLL criterion, with only some small loss in constant terms. 

More recently, Fischer and Ghaffari \cite{ghaffari-lll} described an algorithm running in $2^{\poly(d)+ O(\sqrt{ \log \log n})}$ rounds, under the pLLL criterion $p (e d)^{32} < 1$. Despite the exponential dependence on degree, this algorithm nevertheless can be used to construct a number of combinatorial objects including defective coloring, frugal coloring, and vertex coloring in $2^{O(\sqrt{\log \log n})}$ time for arbitrary degree graphs. 

In this paper, we give new LLL algorithms, that can be faster than those of \cite{ghaffari-lll} and that can be used for higher-degree graphs. The first algorithm gives the following guarantee:
\begin{theorem}
\label{main-lll-thm}
Let $i$ be a positive integer. If $20000 d^8 p \leq 1$, then there is a randomized \LOCAL algorithm to find a configuration
avoiding $\mathcal B$ w.h.p. in time $\exp^{(i)} \bigl( O( \log d + \sqrt{\log^{(i+1)} n}) \bigr)$.
\end{theorem}
For example, with $i = 1$, this runs in $\poly(d) + 2^{O(\sqrt{\log \log n})}$ rounds---this improves the dependence of $d$ by an exponential factor compared to \cite{ghaffari-lll}. For $i = 2$, the algorithm can run in time $2^{2^{O(\sqrt{\log \log \log n})}}$, under the condition that  $d \leq 2^{\sqrt{\log \log \log n}}.$ This makes partial progress toward showing the conjecture of \cite{chang2017time} for the running time of LLL on bounded-degree graphs.

Our generic derandomization method allows us to transform this into a deterministic algorithm, whose running time is stepped up by an exponential factor in $n$:
\begin{theorem}
  \label{xderand-lll-thm}
  Let $i$ be a positive integer. If $20000 d^{8} p \leq 1$, then there is a deterministic \LOCAL algorithm to find a configuration avoiding $\mathcal B$  in time $\exp^{(i)} \bigl(O(\log d + \sqrt{\log^{(i)} n}) \bigr)$.
\end{theorem}

The final LLL algorithm we develop does not make any requirement on the size of $d$, but is not as general in terms of the types of bad-events; it requires that the bad-events satisfy a different property which we refer to as \emph{bounded fragility}. We show the following result:
\begin{theorem}
\label{tr1}
Suppose that every bad-event $B \in \mathcal B$ has fragility at most $F \leq e^{-10} d^{-12}$. Then there is a randomized algorithm to find a configuration avoiding $\mathcal B$ in $2^{O(\sqrt{\log \log n})}$ rounds, w.h.p.
\end{theorem}

This property is satisfied by a number of combinatorial problems such as $k$-SAT and defective vertex coloring. For example, we obtain the following algorithmic applications:
\begin{proposition}
\label{ksat-prop}
If a $k$-SAT formula $\Phi$ has $m$ clauses and every clause intersects at most $d$ others for $d \leq e^{-10} (4/3)^{k/12} \approx 0.00005 \times 1.02426^k$, there is a randomized algorithm to find a satisfying solution to $\Phi$ in $2^{O(\sqrt{\log \log m})}$ rounds.
\end{proposition}

\begin{proposition}
\label{defect-coloring-prop}
Suppose $G$ has maximum degree $\Delta$ and $h \leq \Delta$. There is a randomized algorithm in $2^{O(\sqrt{\log \log n})}$ rounds to find an $h$-defective $k$-coloring with $k = O( \Delta/h )$.
\end{proposition}
We note that \Cref{defect-coloring-prop} was shown by
\cite{ghaffari-lll} using an ad-hoc algorithm based on recoloring
vertices; our main contribution is to derive it as a black-box application of the more general \Cref{tr1}. \Cref{ksat-prop} appears to be the first algorithm for solving $k$-SAT in the distributed setting which makes no restriction on the maximum degree as a function of $n$.

We briefly summarize our algorithmic improvements. Like many previous LLL algorithms, our algorithm is based on a general method for constructing distributed graph algorithms by \emph{graph shattering}\cite{barenboim2016locality,elkin20152delta,gmis,GS17,ghaffari-lll,harris2016coloring,chang2018optimal}. These algorithms have two phases. The first phase satisfies most of the vertices in the graph; the remaining unsatisfied vertices have small connected components. The second phase applies a deterministic algorithm to solve the residual components.

Our derandomization method allows us to convert previous randomized LLL algorithms into deterministic ones; these deterministic algorithms can then be applied for the second phase. This gives us two new, randomized LLL algorithms. The first runs in time $2^{O(\sqrt{\log \log n})}$ for $d \leq 2^{\sqrt{\log \log n}}$. The second runs in time $O(d^2 + \log^* n)$ for LLL instances satisfying a stronger slack condition $p \leq 1/\polylog n$. These can be combined via the bootstrapping methodology of \cite{ghaffari-lll}, giving us the stated runtime bounds.

\subsubsection{The Role of Randomization in the Sequential \textsf{LOCAL} Model}
Besides the above concrete improvements, our derandomization method has implications for the bigger-picture aspects of the role and power of randomization in the \LOCAL model. From \Cref{thm:main_informal}, we show that randomized and deterministic complexities are very close in the $\SLOCAL$ model:
\begin{theorem}\label{thm:SLOCAL_informal} Any randomized \SLOCAL algorithm with locality $r(n)$ for a locally checkable problem can be transformed to a deterministic \SLOCAL algorithm with locality $O(r(n)\log^2 n)$. 
\end{theorem}

As we have discussed, the \SLOCAL model aims to decouple the challenges of locality from those of symmetry breaking.  This was with the intuitive hope that, while the
latter seems to naturally benefit from randomization, locality on its
own should not need randomization. \Cref{thm:SLOCAL_informal} partially validates this intuition, by showing that the \SLOCAL model can only gain relatively small (polylogarithmic) factors from randomness.  That is, in some sense, in the \LOCAL model, randomization
helps mainly with local coordination and symmetry breaking challenges,
but not with the locality challenge. We remark that, if we do care
about logarithmic factors, then a gap appears also in \SLOCAL:

\begin{theorem} The sinkless orientation problem in bounded degree graphs has randomized \SLOCAL locality $\Theta(\log\log\log n)$ and deterministic \SLOCAL locality $\Theta(\log\log n)$.
\end{theorem}  
This exhibits an exponential separation between  randomized and deterministic complexities in the \SLOCAL model, akin to those observed in the \LOCAL model---sinkless orientation in \LOCAL requires $\Omega(\log n)$ rounds deterministically\cite{chang16} and $\Omega(\log\log n)$ rounds randomly\cite{brandt}, and both lower bounds are tight\cite{GS17}. We find it surprising that this gap appears an exponential lower in the \SLOCAL model. 

This also shows that if there is significant separation in the \LOCAL model in the regime of polylogarithmic complexities and higher, it must be for a reason very different than those of \cite{chang16, brandt, GS17} (for sinkless orientation), which show up in regime of sublogarithmic complexities. The latter extend only to sublogarithmic complexities of the \SLOCAL model.

\subsection{Our Contributions, Part II: Limitations of
  Derandomization}
\label{sec:limitations}

In the second part of the paper, we exhibit limitations on derandomization. We present conditional hardness results for some classic and well-studied
distributed problems, including set cover approximation, minimum
dominating set approximation, and neighborhood
covers. Formally, we show that these problems are \polyseqloc-complete, in the framework set
forth by Ghaffari, Kuhn, and Maus~\cite{ghaffari2017complexity}. A
problem $\calP$ is called \polyseqloc-complete if $\calP$ can be
solved deterministically with polylogarithmic locality in the \SLOCAL
model and if a polylog-time deterministic distributed algorithm for
$\calP$ would imply such algorithms for all \polyseqloc-solvable problems.  We provide an informal explanation here; please see
\fullOnly{\Cref{subsec:complete}}
\shortOnly{the full paper}
for a more detailed explanation.

For the above three problems, rather satisfactory polylog-time
randomized \LOCAL algorithms have been known for many years, e.g.,
\cite{bartal97,berger1994efficient,rajagopalan1998primal,dubhashi05,jia02,dominatingset,nearsighted,linial93}. However,
there are no known efficient deterministic algorithms. We show that
devising efficient deterministic algorithms for them may be hard; at
least as hard as some well-known open problems of the area. More
concretely, a polylogarithmic-time deterministic algorithm for any of these problems can be transformed into a
polylogarithmic-time deterministic \LOCAL algorithm for computing a
$(O(\log n), O(\log n))$ network decomposition, which would
consequently imply polylogarithmic-time deterministic \LOCAL
algorithms all \polyseqloc \ problems. The latter class includes, most notably, computing an MIS.  Hence, devising polylogarithmic-time deterministic algorithms for
these problems is at least as hard as doing so for MIS, which
remains a well-known open problem in the area since Linial
explicitly asked for it in 1987\cite{linial1987LOCAL}.

We would like to highlight an implication of the hardness of neighborhood cover problem. This has been a central problem in the study of \emph{local distributed graph
  algorithms} since the work of Awerbuch and
Peleg\cite{Awerbuch-Peleg1990}, closely related to another central
problem, \emph{network decompositions}, introduced by Awerbuch et
al.\cite{awerbuch89}. By classic results of \cite{awerbuch96}, it has
been known that an efficient deterministic network decomposition
algorithm can be transformed into an efficient deterministic
neighborhood cover algorithm. Our result shows for the first time that
the converse is also true: an efficient deterministic neighborhood
cover algorithm can be transformed into one for network decomposition.

Finally, we show \polyseqloc-completeness for the problem of computing an MIS of a type of star
structure in a graph. \fullOnly{Again, we defer the formal description to
\Cref{subsec:complete}. } Significantly, this problem can be solved easily using a simple and
natural greedy method --- with locality $2$ in the \SLOCAL model --- and
yet it is $\polyseqloc$-complete.

\subsection{Organization of the Paper}
\Cref{sec:model}, we formally define the various types of randomized
and deterministic local algorithms and complexity classes that we use throughout the
paper. \Cref{sec:basic-derand-local} then proves the basic
derandomization routine, which will be the main technical tool
for most of our new results. The section also discusses various direct
consequences of the result. In
  \Cref{sec:hypergraph_matching}, we use the
  derandomization procedure to obtain a better deterministic algorithm
  for computing a maximal matching in a hypergraph, and we show some
  consequences of this improved hypergraph matching algorithm. \Cref{lll-sec1} formally proves all the results
regarding the LLL that are discussed in
\Cref{sec:LLLresults}. Finally, in \Cref{sec:SLOCAL}, we discuss the
limitations to derandomization and the completeness results outlined in \Cref{sec:limitations}. In addition, in \Cref{shattering-appendix}, we discuss
some results on graph shattering and how our derandomization results can be used for shattering algorithms.

\hide{
\subsection{Notation}

For a graph $G=(V,E)$ and a subset $X \subseteq V$, we define $G[X]$
to be the vertex-induced subgraph. For any integer $r \geq 1$, $G^r$
is the graph on vertex set $V$ and with an edge between any two nodes
$u$, $v$ that are at distance at most $r$.  Further, $\Delta(G)$
denotes the maximum degree of a graph $G$. Likewise, for a hypergraph
$H = (V,E)$, we define $\Delta(H)$ to be the maximum degree, i.e., the
maximum, over all vertices $v \in V$, of the number of edges $e \in E$
such that $v \in e$. We often write $\Delta$ instead of $\Delta(G)$,
if it is clear from context. 

Unless stated otherwise, $\log(x)$ is the base-two logarithm. We let $\log^{(i)}$ denote the $i$-fold iterated logarithm, e.g. $\log^{(2)} x = \log \log x$. We define $\exp x = 2^x$, and also define $\exp^{(i)} x$ to be the iterated exponentiation to the power $2$, e.g. $\exp^{(2)} x = 2^{2^{x}}$.
}

\input{model}

\input{basic_derandomization}

\shortOnly{\input{hypergraph_matching_sketch}}

\fullOnly{
\input{hypergraph_matching}

\input{LLL}

\input{SLOCAL}
}

\section*{Acknowledgments}

We would like to thank Janosch Deurer, Manuela Fischer, Juho Hirvonen, Yannic
Maus, Jara Uitto, and Simon Weidner for interesting discussions on
different aspects of the paper.

\bibliographystyle{alpha}
\bibliography{references}

\fullOnly{
\appendix
\input{shattering}
}

\end{document}

%% file: model.tex
\section{Model and Definitions}
\label{sec:model}

\paragraph{Notation:}
For a graph $G=(V,E)$ and a subset $X \subseteq V$, we define $G[X]$
to be the vertex-induced subgraph. For any integer $r \geq 1$, $G^r$
is the graph on vertex set $V$ and with an edge between any two nodes
$u$, $v$ that are at distance at most $r$.  Further, $\Delta(G)$
denotes the maximum degree of a graph $G$. 
Likewise, for a hypergraph $H = (V,E)$, we define $\Delta(H)$ to be the maximum degree, i.e., the
maximum, over all vertices $v \in V$, of the number of edges $e \in E$
such that $v \in e$. We define the \emph{rank} of hypergraph $H$ to be the maximum cardinality of any edge. We often write $\Delta$ instead of $\Delta(G)$ or $\Delta(H)$, if it is clear from context.

\paragraph{Distributed Graph Problems:} We deal with \emph{distributed graph problems} of the following form. We are given a simple, undirected graph $G=(V,E)$ that models the network. Initially, each node $v\in V$ gets some private input $x_v$. Each node has a unique ID, which is also part of the initial input. At the end of an algorithm, each node $v\in V$ needs to output a value $y_v$. Let $\vec{x}$ and $\vec{y}$ the vectors of all the inputs and outputs, respectively. An algorithm solving a distributed graph problem needs to guarantee that the triple $(G,\vec{x},\vec{y})$ satisfies the specification of the graph problem. \fullOnly{We assume that whether a given input-output pair satisfies the specification of a graph problem can only depend on the topology of $G$ and it has to be independent of the assignment of IDs to the nodes. For simplicity, we assume that all nodes also know a common polynomial upper bound on the number of nodes $n$.} For a more formal definition, we refer to \cite{ghaffari2017complexity}. We study distributed graph problems in the \LOCAL and the \SLOCAL models that were introduced in \Cref{sec:intro}.

\paragraph{Complexity Classes:} We next define the complexity classes
that we use in the paper. We start by defining the classes for
deterministic algorithms. Throughout, we only consider
complexities as a function of the number of nodes $n$. For a more
general and formal definition, we refer
to \cite{ghaffari2017complexity}.
\begin{description}
\item[{\boldmath$\loc(t(n))$:}] Distributed graph
  problems that can be solved by a \emph{deterministic \LOCAL
    algorithm} with time complexity $t(n)$.
\item[{\boldmath$\seqloc(t(n))$:}] Distributed graph
  problems that can be solved by a \emph{deterministic \SLOCAL
    algorithm} with locality $t(n)$. 
\end{description} 


We distinguish two kinds of randomized algorithms:
Monte Carlo and Las Vegas algorithms. A \emph{distributed
  Monte Carlo} algorithm has a fixed time complexity and it guarantees that the solution solves the graph problem $\mathcal P$ with probability strictly larger than $1-1/n$. For a \emph{distributed Las Vegas} algorithm, we also assume that the time complexity is fixed. However, in addition to the output of the graph problem, each node also outputs a flag $F_v \in \{0, 1 \}$, which serves as an indicator of whether the algorithm failed locally at $v$. If $F_v = 0$ for every node $v$, it is guaranteed that the computed output solves $\mathcal P$. Furthermore, it is guaranteed that $\sum_{v \in V} \E[F_v] < 1$.

We note that this definition of Las Vegas algorithms initially seems to be very different from standard notions for randomized \LOCAL algorithms. In Section~\ref{sec:alternate-def}, we show how this definition is equivalent (up to polylogarithmic factors) with more standard definitions of Las Vegas algorithms from complexity theory. However, we adopt this definition in terms of indicator variables $F_v$ for technical reasons.

\begin{description}
\item[{\boldmath$\randloc(t(n))$:}] Distributed graph problems that can be solved by a \emph{randomized Monte Carlo algorithm} in the \LOCAL model with time complexity $t(n)$.
\item[{\boldmath$\zloc(t(n))$:}] Distributed graph problems that can be solved by a \emph{randomized Las Vegas algorithm} in the \LOCAL model with time complexity $t(n)$.
\item[{\boldmath$\randseqloc(t(n))$:}] Distributed graph problems that can be solved by a \emph{randomized Monte Carlo algorithm} in the \SLOCAL model with time complexity $t(n)$.
\end{description}

We clearly have $\zloc(t(n))\subseteq \randloc(t(n))$. In addition, if the validity of a solution to a graph problem can be locally checked by exploring the $d$-neighborhood of each node (cf.\ \cite{fraigniaud13} for a formal definition of local decision problems), then a $\randloc(t)$ algorithm yields a $\zloc(t + d)$ algorithm. We will show in \Cref{sec:SLOCAL2} an example of a problem in $\randloc(0)$ but not in $\zloc(o(\sqrt{n}))$.

We often think of an \LOCAL algorithm as \emph{efficient} if it has time complexity at most $\polylog n$. We therefore define {\boldmath$\polyloc$}, {\boldmath$\polyseqloc$}, {\boldmath$\polyrandloc$}, and {\boldmath$\polyzloc$} as respectively $\loc(\polylog n)$, $\seqloc(\polylog n)$, $\randloc(\polylog n)$, and $\zloc(\polylog n)$.

\hide{
\fabian{The following should go to the subsection proving the completeness results.}
\paragraph{Locality-Preserving Reductions and Completeness:} 
To relate different problems to each other, we use the notion of
\emph{locality-preserving reductions} as introduced as formally
defined in \cite{ghaffari2017complexity}. A distributed graph problem
$\mathcal{A}$ is called \emph{polylog-reducible} to a distributed
graph problem $\mathcal{B}$ iff a $\polylog n$-time deterministic
\LOCAL algorithm for $\mathcal{A}$ (for all possible $n$-node graphs)
implies a $\polylog n$-time deterministic \LOCAL algorithm for
$\mathcal{B}$. In addition, a distributed graph problem $\mathcal{A}$
is called \polyseqloc-complete if a) the problem $\mathcal{A}$ is in
the class \polyseqloc and b) every other distributed graph problem in
\polyseqloc\ is polylog-reducible to $\mathcal{A}$. As a consequence,
if any \polyseqloc-complete problem can be solved in $\polylog n$
deterministic time in the \LOCAL model, then we have
$\polyloc = \polyseqloc$ and thus every problem that can be solved
deterministically with $\polylog n$ locality in the \SLOCAL model can
also be solved in $\polylog n$ time deterministically in the \LOCAL
model. A particularly important problem that was shown to be
\polyseqloc-complete in \cite{ghaffari2017complexity} is the problem
of computing a $(\polylog n,\polylog n)$-network decomposition:
}

\paragraph{Network Decomposition:} This graph structure plays a central role in this paper and \LOCAL algorithms.
\begin{definition}[Network Decomposition]
  \label{def:decomposition}\emph{\cite{awerbuch89}}
  A $\big(d(n),c(n)\big)$-decomposition of an $n$-node
  graph $G=(V,E)$ is a partition of $V$ into $c(n)$ classes $V = V_1 \cup \dots \cup V_{c(n)}$, with the property that each induced subgraph $G[V_i]$ has connected components of diameter at most $d(n)$. The sets $V_i$ are referred to as \emph{colors} and the connected components of each $G[V_i]$ are referred to as \emph{clusters}.
\end{definition}

It was shown in \cite{Awerbuch-Peleg1990,linial93} that every graph
has an $\big(O(\log n),O(\log n)\big)$-decomposition. Further, as
shown in \cite{ghaffari2017complexity}, the algorithm of
\cite{Awerbuch-Peleg1990,linial93} directly leads to an
$\SLOCAL(O(\log^2 n))$-algorithm for computing an
$\big(O(\log n),O(\log n)\big)$-decomposition. In \cite{linial93},
Linial and Saks describe a $O(\log^2 n)$-time distributed
Las Vegas algorithm to compute a
$\big(O(\log n),O(\log n)\big)$-decomposition. Further, in
\cite{ghaffari2017complexity}, it was shown that the problem of
computing a $(d(n),c(n))$-decomposition is \polyseqloc-complete for $d(n),c(n) = O(\log^k n)$ and any constant $k\geq 1$. As a consequence, $\polyseqloc \subseteq \polyzloc$. The best deterministic distributed
algorithm to compute a $\big(O(\log n),O(\log n)\big)$-decomposition, due to Panconesi and Srinivasan \cite{panconesi95}, has time
complexity $2^{O(\sqrt{\log n})}$.


%% file: basic_derandomization.tex
\section{Basic Derandomization of Local Algorithms}
\label{sec:basic-derand-local}

In the previous section, we have defined two classes of randomized
distributed algorithms. Our derandomization technique only applies to Las Vegas algorithms; however, this is only a slight restriction, as most Monte Carlo graph algorithms (including, for instance, all locally-checkable problems), can be converted into Las Vegas algorithms. The formal statement of our basic derandomization
result is given as follows:

\begin{theorem}\label{thm:basicderandomization}
  Let $\calP$ be a  graph problem  which has a
  $\zloc(r)$ algorithm $A$.  When running $A$ on a graph
  $G=(V,E)$, for each $v\in V$ let $R_v$ be the private random bit
  string used by node $v$. Then there is a deterministic
  $\seqloc[2r]$-algorithm that assigns values to all $R_v$ such that
  when (deterministically) running $A$ with those values, it
  solves $\calP$.
\end{theorem}
\begin{proof}
Consider a randomized run of the Las Vegas algorithm $A$. In it, every node $v$ sets a flag $F_v$. Let us further define the random variable
  $F:=\sum_{v\in V} F_v$ such that $F \geq 1$ iff $A$ fails to
  compute a solution for $\calP$. By definition, we have $\E[F] = \sum_{v\in V} \E[F_v] < 1$.

  We now construct a deterministic \seqloc-algorithm $A'$ from $A$ via the method of conditional expectation. Suppose that $A'$ processes the nodes in some arbitrary order $v_1, \dots, v_n$. When processing node
  $v_i$, $A'$ needs to fix the value of $R_{v_i}$. The algorithm $A'$ will choose values $\rho_i$ for each  $i\in\set{1,\dots,n}$ to ensure that
  \begin{equation}\label{eq:condexpectation}
    \E\big[F \,\big|\, R_{v_1}=\rho_1, \dots, R_{v_i}=\rho_i \big] \leq \E[F]<1.
  \end{equation}
  After processing all $n$ nodes, all values $R_{v_i}$ are set to a
  fixed value and for $i=n$, the conditional expectation of $F$ is simply the final value of $\sum_v F_v$ when running $A$ with these values for
  $R_{v_i}$. Because $\E[F]<1$, \Cref{eq:condexpectation} thus implies
  that $\E[F\,|\,R_{v_1}=\rho_1,\dots,R_{v_n}=\rho_n]<1$ and thus the
  algorithm succeeds in solving $\calP$. It remains to show how to set $\rho_i$ to satisfy \Cref{eq:condexpectation}, and such that when
  processing node $v_i$, the \seqloc-algorithm $A'$ only needs to
  query the $2r$-neighborhood of $v_i$. 

Now suppose that the values of $\rho_1,\dots,\rho_{i-1}$ are already given such
  that
  $\E[F\,|\,R_{v_1}=\rho_1,\dots,R_{v_{i-1}}=\rho_{i-1}]\leq \E[F]$, and let $S = \{v_1, \dots, v_{i-1} \}$.
  \Cref{eq:condexpectation} can clearly be satisfied by choosing
  $\rho_i$ to minimize
  $\E[F\,|\,R_{v_1}=\rho_1,\dots,R_{v_i}=\rho_i]$. Furthermore, the output of any node $v$ in the distributed algorithm, and hence also the flag $F_v$,
  only depends on the initial state of the $r$-hop neighborhood of $v$. Because the values of $R_u$ are all independent random variables, this implies
  
  \begin{equation}
    \label{eq:condindependence}
    \E \Bigl[F_v\,\big|\,\bigwedge_{u\in S} (R_u=\rho_u) \Bigr] =
    \E \Bigl[F_v\,\big|\,\bigwedge_{\substack{u  \in S  \\ d_G(u,v)\leq r}} (R_{u}=\rho_u) \Bigr].
  \end{equation}
  We can therefore fix the value of $\rho_i$ as follows.
{\allowdisplaybreaks
  \begin{eqnarray*}
    \rho_i 
    & = & \arg\min_{\rho}
          \E\big[F\,\big|\,R_{v_1}=\rho_1,\dots,R_{v_{i-1}}=\rho_{i-1},R_{v_i}=\rho]\\
    & = & \arg\min_{\rho} \sum_{v\in V}
          \E\big[F_v\,\big|\,R_{v_1}=\rho_1,\dots,R_{v_{i-1}}=\rho_{i-1},R_{v_i}=\rho]\\
    & = & \arg\min_{\rho} \sum_{v\,:\,d_G(v_i,v)\leq r} 
          \E \Bigl[F_v\,\big|\,R_i=\rho \land \bigwedge_{\substack{ j< i \\ d_G(v,v_j)\leq r}} (R_{v_j}=\rho_j) \Bigr]
  \end{eqnarray*}
}
  The last equation follows because for $v$ at distance more than
  $r$ from $v_i$, by \Cref{eq:condindependence},
  $\E\big[F_v\,\big|\,R_{v_1}=\rho_1,\dots,R_{v_{i-1}}=\rho_{i-1},R_{v_i}=\rho]$
  does not depend on the value of $\rho$. Thus, when determining the
  value of $\rho_i$, $A'$ needs to evaluate
  conditional expectations of $F_v$ for all $v$ within distance at most
  $r$ from $v_i$. In order to do this, it is sufficient to read the
  current state of the $2r$-neighborhood of $v_i$.
\end{proof}

The key tool to turn an \SLOCAL algorithm back into a distributed
algorithm is the computation of network decompositions (cf.\
\Cref{def:decomposition}). The formal statement of how to use network
decompositions is given by the following proposition, which was
implicitly proven in \cite{ghaffari2017complexity}. \shortOnly{Please see the full paper for
complete proofs.}

\begin{proposition}\label{prop:StoL}
  Suppose we are provided a $(d(n),c(n))$-network
  decomposition of $G^{r(n)}$. If a graph problem $\cal P$ has a $\seqloc[r(n)]$ algorithm (respectively, $\randseqloc[r(n)]$ algorithm), then  $\mathcal P$ can be solved in
  $O\big( (d(n)+1) c(n) r(n) \big)$ rounds by a deterministic (respectively, randomized) \LOCAL algorithm.
  \end{proposition}
\fullOnly{
\begin{proof}
  We use the network decomposition of $G^{r(n)}$ to run the \SLOCAL
  algorithm $A$ in the distributed setting.   We run $A$ by
  processing the nodes according to the increasing lexicographic order
  given by the color of a node and the ID of the node. A
  network decomposition of $G^{r(n)}$ guarantees that nodes of
  different clusters of the same color are at least $r(n)+1$ hops away
  from each other. Hence, when processing the nodes of a given color,
  different clusters cannot interfere with each other and we can
  locally simulate the execution of $A$ in each cluster. In the
  \LOCAL model, this local simulation can be done in
  $r(n)\cdot (d(n)+1)$ rounds for each cluster. We then have to iterate over the $c(n)$ colors.
\end{proof}
}

We next show how we can convert back and forth between \ZLOCAL,
\SLOCAL, and \LOCAL algorithms.

\begin{proposition}\label{prop:ZtoL} 
Any algorithm $A \in \seqloc[r] \cup \zloc[r]$ to solve a graph problem
    $\mathcal P$ on $G$ yields a deterministic \LOCAL algorithm in $\min\set{r\cdot 2^{O(\sqrt{\log n})}, O\big(r\cdot(\Delta(G^{2r}) +
      \log^*n)\big)}$ rounds.
\end{proposition}
\begin{proof}
  We first consider $A \in \seqloc[r]$.  To get the first bound, we compute an
  $\big(O(\log n), O(\log n)\big)$-decomposition of $G^r$ in time
  $r\cdot 2^{O(\sqrt{\log n})}$ by using the algorithm of
  \cite{panconesi95}. By \Cref{prop:StoL} we use this to simulate
  $A$ in $O(r \log^2 n)\leq r2^{O(\sqrt{\log n})}$ rounds.  For
  the second bound, we compute a $\Delta(G^{2r}+1)$-coloring of $G^{2r}$ in
  time $O\big(r(\Delta(G^{2 r}) + \log^*n)\big)$ by using the algorithm of
  \cite{BEK15}. This can be viewed as a
  $(0,\Delta(G^{2r})+1)$-decomposition of $G^{2r}$.  By \Cref{prop:StoL} we
  use this to simulate $A$ in $O(r \Delta(G^{2r}))$ rounds.

  For $A \in \zloc[r]$, \Cref{thm:basicderandomization} gives an $\seqloc[2 r]$ algorithm to determine a set of random bits
  that make $A$ succeed; we can turn this into a deterministic
  \LOCAL algorithm by using part (1). We can then simulate $A$
  (which is at this point a deterministic algorithm) in $O(r)$
  additional rounds.
\end{proof}

\fullOnly{
Assume that we are given a \ZLOCAL algorithm $A$ for a given graph
problem $\calP$.
The basic derandomization technique in \Cref{thm:basicderandomization}
only shows how to get an \SLOCAL algorithm to determine the private
randomness $R_v$ of the \ZLOCAL algorithm $A$ for every node
$v$. We can extend this to get a deterministic \SLOCAL algorithm for
the original graph problem $\calP$.
}

\begin{proposition}\label{prop:ZtoS}
  $\zloc[r] \subseteq \seqloc[4 r]$.
\end{proposition}
\fullOnly{
\begin{proof}
  Let $A$ be a $\zloc[r]$ algorithm for a graph problem $\cal P$.
  \Cref{thm:basicderandomization} shows that we can determine a
  setting for the random bits which causes $A$ to succeed. Once
  these are determined, we can execute $A$ as
  a deterministic (and hence \SLOCAL) algorithm with locality
  $r$. Thus, we can solve $\cal P$ by composing two \SLOCAL algorithms,
  with localities $r_1 = 2 r$ and $r_2 = r$, respectively. As shown in
  Lemma 2.3 of \cite{ghaffari2017complexity}, this composition can be
  realized as a \emph{single} $\seqloc[r']$ algorithm for 
  $r' = r_1 + 2 r_2 = 4 r$.\footnote{Lemma 2.3 of the conference paper
    \cite{ghaffari2017complexity} only claims a locality of
    $r'=2(r_1+r_2)$. The tight bound and a full proof appear as Lemma
    2.2 in the full version of \cite{ghaffari2017complexity}, which is
    available at \url{https://arxiv.org/abs/1611.02663}.}
\end{proof}
}

Using the same
basic techniques as in the above propositions, we can also prove
\Cref{thm:SLOCAL_informal}.

\fullOnly{
\begin{proof}[\emph{\textbf{Proof of \Cref{thm:SLOCAL_informal}}}]
  Consider an $n$-node graph $G$ and an
  $\randseqloc(r(n))$ algorithm $A$ for a graph problem $\mathcal P$. We compute an $\big(O(\log n),O(\log n)\big)$-network-decomposition of
  $G^{r(n)}$ in $O(r(n)\log^2 n)$ round using the randomized algorithm of \cite{linial93}. 
   \Cref{prop:StoL} allows us to transform $A$ into a randomized distributed algorithm with time complexity
  $O(r(n)\log^2 n)$. Because $\mathcal P$ is locally checkable,  \Cref{prop:ZtoS} converts this randomized distributed algorithm
  into a deterministic $\SLOCAL[O(r(n)\log^2 n)]$ algorithm.
\end{proof}
}

\begin{corollary}
  \label{cor-zs}
$\polyzloc = \polyseqloc$.
\end{corollary}
\begin{proof}
The inclusion $\polyseqloc \subseteq \polyzloc$, already mentioned, was shown in \cite{ghaffari2017complexity}. The inclusion $\polyseqloc \subseteq \polyzloc$ is Proposition~\ref{prop:ZtoS}.
\end{proof}

\subsection{Alternative Definition of Las Vegas Algorithms}
\label{sec:alternate-def}
The definition of Las Vegas algorithms is somewhat non-standard, and does not exactly correspond to the definitions of Las Vegas algorithms in e.g., complexity theory. One significant complication is that in the computational setting, it is possible to check if the algorithm has terminated at some point and, if necessary, restart from scratch; this is not possible in general for \LOCAL algorithms. We next show how our definition relates to other notions of zero-error randomization.

\begin{definition}[Zero-error distributed algorithm]
  A \emph{zero-error distributed algorithm} is a randomized \LOCAL algorithm $A$ for a graph problem $\calP$ such that when all nodes terminate, $A$ always computes a correct solution for $\calP$.
\end{definition}
\begin{definition}[Exponentially-convergent algorithm]
  An $r$-round \emph{exponentially-convergent algorithm} has the property that for any integer $k \geq 1$, the algorithm terminates within $r k$
  rounds with probability at least $1 - n^{-k}$.
\end{definition}
\begin{definition}[Expected-$r$-round algorithm]
  An \emph{expected-$r$-round algorithm} has the property that the expected time before $A$ terminates is at most $r$.  
\end{definition}

These definitions try to captures the intuition that the algorithm is converging to a solution, and terminates quickly on average. Propositions~\ref{prop:ZtoL2}  show that these different definitions of Las
Vegas algorithms  are all equivalent up to polylogarithmic factors.

\begin{proposition}\label{prop:ZtoL2}
  For a graph problem $\mathcal P$, the following are equivalent:
  \begin{enumerate}
  \item $\mathcal P \in \polyzloc$
    \item $\mathcal P \in \polyseqloc$
  \item $\mathcal P$ has a zero-error $\polylog(n)$-round exponentially-convergent distributed algorithm.
  \item $\mathcal P$ has a zero-error expected-$\polylog(n)$-round distributed algorithm.
  \end{enumerate}
\end{proposition}
\begin{proof}
  The equivalence of $(1)$ and $(2)$ is Corollary~\ref{cor-zs}.
  
  $(2) \Rightarrow (3)$. Let $A$ be a $\seqloc[r]$ algorithm for $\mathcal P$; we will construct an $O(r \log^2 n)$-round zero-error exponentially convergent distributed algorithm.   Consider the following process. We first run the randomized,  distributed network decomposition algorithm of \cite{linial93} for $O(\log^2 n)$ rounds, hoping to produce a  $\big(O(\log n),O(\log
  n)\big)$ network decomposition of $G^r$. We form a vertex set $X$, which are the vertices $v$ such that every neighbor in its color class has distance $O(r \log n)$ from $v$. We iterate over the $O(\log n)$ color classes; within each, the vertices in $X$ simulate $A$ with respect to an arbitrary ordering of $X$, and terminate. Vertices which are not in $X$ sit idle during this process.

  After this process is finished, we repeat the process on the residual graph $V - X$, and so on. Note that, there is a probability of at least $1 - 1/n$ that a given repetition of this process has $X = V$, and so every node terminates. This probability is conditional on all the previous repetitions having failed. Since each repetition runs for $O(r \log^2 n)$ time, this satisfies the conditions of the zero-error exponentially-convergent distributed algorithm.

  $(3) \Rightarrow (4)$. It is immediate that any zero-error  $r$-round exponentially-convergent distributed algorithm is also an zero-error expected-$2r$-round distributed algorithm.

  $(4) \Rightarrow (1)$. Let $A$ be a zero-error
    expected-$r(n)$-round distributed algorithm to solve $\mathcal P$ for $r(n) = \polylog n$. Further, assume that we are given an $n$-node input
    graph $G=(V,E)$ where all nodes have unique IDs from the range
    $\set{0,\dots,N-1}$ for $N=\poly(n)$. Before running $A$, the
    nodes of $G$ randomize their IDs as follows. Each node $v\in V$
    chooses an integer $\alpha_v$ uniformly at random from $\set{0,\dots,n^4}$
    and sets its new ID $y_v$ to $y_v:=\alpha_v\cdot N + x_v$,
    where $x_v\in \set{0,\dots,N-1}$ is the original ID of $v$. Note
    that the new random IDs are still unique. Instead of running $A$
    on $G$ with the original IDs, we run $A$ on $G$ with the random
    IDs $y_v$ and with a fake value of $n$, namely $n'=n^2$, and
    terminating after $3r(n')=\polylog n$ rounds. Every node $v$ that
    has terminated by that time sets $F_v:=0$, all other nodes set
    $F_v:=1$.

    Because the algorithm is run on a graph with at most
    $n'$ nodes and unique IDs, if $F_v=0$ for all $v\in V$, then
    algorithm solves $\mathcal{P}$. (Recall that we assume that
    the validity of a solution to a distributed graph problem is
    independent of the assignment of the IDs.) For a node $v \in V$, it remains to show that
    $\Pr(F_v=1)<1/n$ for each $v\in V$ and thus
    $\sum_{v\in V} \E[F_v]<1$. 
    
    To show this, let $G'$ be a $n'$-node graph that
    consists of $n$ disjoint copies of $G$, where in each copy the nodes choose random IDs in the same way as above. Running the algorithm on $G'$ is equivalent to running it on
    $n$ independent copies of $G$. If
    the all the nodes of $G'$ have unique IDs, then $A$ terminates on all nodes of $G'$ in expected time $r(n')$. Hence,
    after running the algorithm for $3r(n')$ rounds, by Markov's
    inequality, the probability that not all $n'$ nodes have
    terminated is at most $1/3$. The probability that all $n'$ IDs are
    unique is at least
    $\big(1-\frac{n}{n^4}\big)^{n^2}\geq 4^{-1/n}$. Letting $\mathcal{E}$
    be the event that all nodes of $G'$ terminate after $3r(n')$
    rounds and $\mathcal{U}$ be the event that all IDs in $G'$ are
    unique, we have
    \[
    \Pr(\mathcal{E}) \geq \Pr(\mathcal{E}\cap\mathcal{U}) =
    \Pr(\mathcal{E}|\mathcal{U})\cdot\Pr(\mathcal{U})\geq
    \tfrac{2}{3}\cdot 4^{-1/n}\geq \frac{1}{2},
    \]
    for $n\geq 5$. 
    
    Let $v\in V$ be a node of $G$ and assume that the
    probability that $v$ terminates by time $3r(n')$ is $p$. Because
    the $n$ copies are independent, we have $p^n\geq 1/2$ and
    therefore $p\geq 2^{-1/n}>1-1/n$. We therefore have
    $\sum_{v\in V} \E[F_v]<1$.
\end{proof}

%% file: hypergraph_matching.tex
\section{Deterministic Hypergraph Maximal Matching}
\label{sec:hypergraph_matching}
We consider a hypergraph $H$ on $n$ nodes, maximum degree at most $\Delta$, and rank at most $r$. Although the \LOCAL model is defined for graphs, there is a very similar model for hypergraphs: in a single
communication round, each node $u$ can send a message to each node $v$
for which $u$ and $v$ are contained in a common hyperedge. The
objective of the section is to compute a maximal matching of $H$, that
is, a maximal set of pairwise disjoint hyperedges.

Our construction is based on a method of partitioning the hyperedges of a hypergraph $H$ into two classes, so hyperedges of each node are roughly split into two equal
parts, which we refer to as \emph{hypergraph degree splitting}. This degree splitting procedure uses the
derandomization lemma \Cref{thm:basicderandomization} as its core
tool. 

\begin{definition}[Hypergraph Degree Splitting]\label{def:HG_splitting}
  Let $H=(V,E)$ be a hypergraph and let $\delta\geq 1$ and $\eps>0$ be
  two parameters. An $(\eps,\delta)$-degree splitting of $H$ is a
  coloring of the hyperedges with two colors red and blue such that
  for each node $v \in V$ of degree $\deg_H(v)\geq \delta$, at least
  $\frac{1-\eps}{2}\cdot\deg_H(u)$ of the hyperedges of $u$ are
  colored red and at least $\frac{1-\eps}{2}\cdot\deg_H(u)$ of the
  hyperedges of $u$ are colored blue.
\end{definition}

\begin{lemma}\label{lemma:HG_splitting}
  Let $H$ be an $n$-node hypergraph with maximum degree at most
  $\Delta$ and rank at most $r$. Then, for every $\eps>0$, there is a deterministic
  $O(r \log(n \Delta) / \eps^2 )$-round \LOCAL algorithm to
  compute an $\big(\eps, \frac{8 \ln (n \Delta)}{\eps^2}\big)$-degree
  splitting of $H$.
\end{lemma}
\begin{proof}
  Define $\delta := \frac{8}{\eps^2}\cdot \ln(n\Delta)$. 
 We assume that $n\geq n_0$ for a sufficiently large constant
  $n_0\geq 1$ (for constant $n$, the statement of
  the lemma is trivial).  As a first step, we reduce the problem on
  $H$ to the hypergraph splitting problem on a low-degree hypergraph
  $H'$. To construct $H'$, we divide each node of $H$ of
  degree $\geq 2\delta$ into virtual nodes, each of degree
  $\Theta(\delta)$. More specifically, for each node $u\in V$, we
  replace $u$ by
  $\ell_u := \max\set{1,\left\lfloor\deg_H(u)/\delta\right\rfloor}$
  virtual nodes $u_1,\dots,u_{\ell_u}$ and we assign each of the
  hyperedges of $u$ to exactly one of the virtual nodes
  $u_1,\dots,u_{\ell_u}$. If $\ell_u>1$, we divide the hyperedges that
  each virtual node $u_i$ has degree at least $\delta$ and less than
  $2\delta$. The hypergraph $H'$ has at most $n \lceil \Delta/ \delta \rceil$ vertices, maximum degree less than $2\delta$ and an $(\eps,\delta)$-degree splitting
  of $H'$ immediately gives an $(\eps,\delta)$-degree splitting of
  $H$.

  We thus need to show how to efficiently compute an
  $(\eps,\delta)$-degree splitting of the low-degree hypergraph
  $H'$. Instead of directly working on $H'$, it is more convenient to
  define the algorithm on its incidence graph; this is the bipartite graph
  $B=(U_B\cup V_B, E_B)$ which one node in $U_B$ for
  every node $u$ of $H'$ and it has one node in $V_B$ for every
  hyperedge of $H'$. A node $u\in U_B$ and a node $v\in V_B$ are
  connected by an edge of $B$ if and only if the node of $H'$ corresponding
  to $u$ is contained in the hyperedge of $H'$ corresponding to
  $v$. Clearly, any $r$-round computation on $H'$ can be simulated in
  $B$ in at most $2r$ rounds. An $(\eps,\delta)$-splitting of $H'$ now
  corresponds to a red/blue-coloring of the nodes in $V_B$ such that
  every node of degree $d\geq \delta$ in $U_B$ has at least
  $(1-\eps)\frac{d}{2}$ red and at least $(1-\eps)\frac{d}{2}$ blue
  neighbors in $V_B$.

  We first claim that such a red/blue
  coloring of $B$, and thus an
  $(\eps,\delta)$-degree splitting of $H'$, can be computed by a
  trivial Las Vegas algorithm. Each node in $V_B$ colors itself red or
  blue independently with probability $1/2$. For a node $u\in U_B$,
  let $X_u$ and $Y_u$ be the number  red and blue neighbors in $V_B$
  after this random coloring step. If the degree of $u$ is
  less than $\delta$, the coloring does not need to satisfy any
  condition. Otherwise, we know that the degree of $u$ is in
  $[\delta,2\delta)$. Therefore, by the Chernoff bound:
      \[
    \Prob{X_u < (1-\eps)\cdot\frac{\deg_{H'}(u)}{2}} \leq
    e^{-\eps^2\deg_{B}(u)/4} \leq e^{-\eps^2\delta/4} = \frac{1}{(n\Delta)^2}
  \]
  The corresponding bound for $Y_u$ is obtained in the same way.

As $|U_B| \leq n \Delta$, the expected number of failing nodes is therefore at most
  \[
  \Prob{ \bigvee_{u \in U_B}
    \max\set{X_u,Y_u}<(1-\eps)\cdot\frac{\deg_{H'}(u)}{2}}
  < 2\cdot|U_B|\cdot \frac{1}{(n\Delta)^2}
  \leq 2 n \Delta \cdot \frac{1}{(n\Delta)^2} < 1
  \]
  
  This procedure can be computed in $0$ rounds (without communicating), and the correctness can be verified in $1$ round. So it is 
  a $\zloc[1]$ algorithm. Since $B^2$ has maximum degree $O(r\log(n\Delta)/\eps^2)$, Proposition~\ref{prop:ZtoL} shows that it can be executed as a deterministic algorithm in $O(r \Delta(B^2) + r \log^*n) \leq O( r \log(n \Delta)/\epsilon^2)$ rounds. 
\end{proof}

\begin{lemma}
  \label{match-prop} 
Suppose we are given a hypergraph $H = (V,E)$ of maximum degree $\Delta$ as well as an explicitly provided subset $U \subseteq V$ of its vertices and a parameter $\delta$ such that $d(u) \geq \delta$ for every $u \in U$. Then there is a deterministic \LOCAL algorithm in $O(r \Delta \log r + \log r \log^* n)$ rounds to compute a matching $M \subseteq E$ such that
$$
\sum_{e \in M} |e \cap U| \geq \Omega\left( \frac{|U| \delta}{r \Delta} \right).
$$
\end{lemma}
\begin{proof}
Let $k = \lceil \log_2 r \rceil$ and for $i = 0, \dots, k$ let $E_i \subseteq E$ denote the set of edges $e$  with $2^i \leq |e \cap U| < 2^{i+1}$. We will sequentially construct matchings $M_k, \dots, M_0$, as follows. At stage $i$, we define $E'_i$ to be the set of edges $e \in E_i$ which do not intersect with any edge $f \in M_{i+1} \cup \dots \cup M_k$; we find a maximal matching $M_i$ of the hypergraph $H_i = (V_i, E'_i)$. We then finish by outputting $M = M_0 \cup M_1 \cup \dots \cup M_k$.

This construction ensures that, for any $i = 0, \dots, k$, the matching $M_i \cup M_{i+1} \cup \dots \cup M_k$ is a maximal matching of the hypergraph
$(V, E_i \cup E_{i+1} \cup \dots \cup E_k)$. Since the line graph of $H$ has maximum degree $s = r \Delta$, this shows that
$$
|M_i| + \dots + |M_k| \geq \frac{|E_i| + \dots + |E_k|}{s}.
$$

For any $u \in U$, let $d_i(u)$ denote the number of edges $e \in E_i$ with $u \in e$. By double-counting, we have $|E_i| \geq \frac{\sum_{u \in U} d_i(u)}{2^{i+1}}$  and $\sum_i d_i(u) \geq \delta$ for every $u \in U$. We now compute:
{\allowdisplaybreaks
\begin{align*}
\sum_{e \in M} |e \cap U| &= \sum_{i=0}^k \sum_{e \in M_i} |e \cap U| \geq \sum_{i=0}^k 2^i |M_i| \geq \sum_{i=0}^k 2^{i-1} \sum_{j=i}^k |M_j| \\
&\geq \sum_{i=0}^k 2^{i-1}/s \sum_{j=i}^k |E_j| = \sum_{j=0}^k 2^j |E_j|/s = \sum_{j=0}^k 2^j \frac{\sum_{u \in U} d_j(u)}{2^{j+1} s} \\
&= \sum_{j=0}^k \frac{\sum_{u \in U} d_j(u)}{2 s} \geq \frac{|U| \delta}{2 s}
\end{align*}
}

This procedure goes through $O(\log r)$ stages. In each stage $i$, we compute a maximal matching of the hypergraph $H_i$, whose line graph has maximum degree $s$; so this can be achieved in $O( s + \log^* (n \Delta) )$ time per stage. (We note that $\Delta \leq 2^n$, so $\log^* \Delta \leq O(\log^* n)$).
\end{proof}
\begin{lemma}
  \label{split2}
  Given a rank-$r$ hypergraph $H = (V,E)$ of maximum degree
  $\Delta = \Omega(\log n)$, there is a deterministic
  $O(r \log r \log n + r \log(n \Delta) \log^3 \Delta)$-round \LOCAL algorithm to find a
  matching $M$ of $H$ with the following property: if $U_+$ denotes
  the set of vertices of $H$ of degree at least $\Delta/2$, then the
  edges of $M$ contain at least an $\Omega(1/r)$ fraction of the
  vertices of $U_+$.
\end{lemma}
\begin{proof}
The problem is trivial when $n \leq O(1)$ so we assume without loss of generality that $n\geq n_0$ for a sufficiently large constant.
  
We reduce the degree of $H$ by repeatedly applying the
  hypergraph degree splitting of \Cref{lemma:HG_splitting}.
  We begin by setting $E_0 = E$. For an integer $t \geq 1$, we will define parameters $\eps_1,\dots,\eps_t>0$ and
  $\delta_1,\dots,\delta_t\geq 1$, and  construct edge sets
  $E_1,\dots, E_t$ such that $E_{i}\subseteq E_{i-1}$ as follows. For
  each $i\in \set{1,\dots,t}$, we use \Cref{lemma:HG_splitting} to
  compute an $(\eps_i,\delta_i)$-splitting of the hypergraph $(V,E_{i-1})$; we define $E_i$ to be the resulting hyperedges that are colored red.
  
  We choose the parameters as follows:
  \begin{eqnarray*}
    t & := & \max\set{1, \left\lfloor\log\Delta - \log\log n - 14\right\rfloor} \\
     \delta_i
      & := & \frac{\Delta}{2^{i+1}}\\
    \eps_i
      & := &
             \max\set{\frac{1}{4\log\Delta}\, ,\, \sqrt{\frac{16\ln(n\Delta/2^{i-1})}{\Delta/2^i}}}
  \end{eqnarray*}
  
  Let us first observe that  $\sum_{j=1}^t\eps_j\leq 1/2$. Assuming $n\geq n_0$
  for a sufficiently large constant $n_0$, we have
  {\allowdisplaybreaks
  \begin{eqnarray*}
    \sum_{i=1}^t \eps_i 
    & \leq & \sum_{i=1}^t\frac{1}{4\log\Delta}  + 
             \sum_{i=1}^t
             \sqrt{\frac{16\ln(n\Delta/2^{i-1})}{\Delta/2^i}}\\
    & \stackrel{(t\leq\log\Delta)}{\leq} & 
                                           \frac{1}{4} + \sum_{j=\log\Delta-t}^\infty
                                           \sqrt{\frac{16\big(\ln n + (j+1)\ln 2\big)}{2^j}}\\
    & \stackrel{(n\geq n_0)}{\leq} & \frac{1}{4} + \sum_{s=14}^\infty
                                     \sqrt{\frac{32\ln n + s}{2^s\cdot
                                     \ln n}}  \stackrel{(n\geq n_0)}{\leq} \frac{1}{4} + \sum_{s=14}^\infty
                                     \sqrt{\frac{32}{1.8^s}}\ <\ \frac{1}{2}.
  \end{eqnarray*}
  }
  Let $H_i = (V,E_i)$ and $\Delta_i = \Delta(H_i)$. We show by induction that $\Delta_i \leq(1+2\sum_{j=1}^i \eps_j)\Delta/2^i\leq \Delta/2^{i-1}$, and any node in $U_+$ has degree at least
  $(1-\sum_{j=1}^i\eps_j)\Delta/2^{i+1}\geq \Delta/2^{i+2}$ in $H_i$. For $i=0$, these bounds clearly hold
  because by definition, each node of $H$ has degree at most
  $\Delta$ and each node in $U_+$ has degree at least $\Delta/2$ in
  $H = H_0$. 
  
  For the induction step, we first show that for each $i$, we
  can apply the degree splitting algorithm of
  \Cref{lemma:HG_splitting} to compute an $(\eps_i,\delta_i)$-degree
  splitting of $H_{i-1}$; specifically, we need to show that $\delta_i \geq
  8\ln(n\Delta_{i-1})/\eps_i^2$. By induction hypothesis, we have
  \[
  \frac{8\ln(n\Delta_{i-1})}{\eps_i^2} \leq
  \frac{8\ln(n\Delta/2^{i-1})}{\frac{16\ln(n\Delta/2^{i-1})}{\Delta/2^i}}
  = \frac{\Delta}{2^{i+1}} = \delta_i.
  \]
  
  By
  induction hypothesis, each node in $U_+$ has degree at least
  $\delta_i$ in $H_{i-1}$. The minimum degree of any node of $U_+$ in $H_i$ is
  therefore at $(1-\eps_i)/2$ times the the minimum degree of any node
  of $U_+$ in $H_{i-1}$. Thus, the minimum degree of any node of $U_+$
  in $H_i$ is at least
  \[
  \Bigl(1-\sum_{j=1}^{i-1}\eps_j \Bigr)\frac{\Delta}{2^i}
    \cdot\frac{1-\eps_i}{2} =
    \Bigl(1-\sum_{j=1}^{i-1}\eps_j \Bigr)(1-\eps_i)\cdot\frac{\Delta}{2^{i+1}}
      \geq
    \Bigl(1-\sum_{j=1}^{i}\eps_j \Bigr)\cdot\frac{\Delta}{2^{i+1}}
    \geq
    \frac{\Delta}{2^{i+2}}.
  \]

  Finally, note that maximum
  degree of $H_i$ is at most $(1+\eps_i)/2$ times the maximum degree
  of $H_{i-1}$. So
  \[
  \Delta_i \leq  \Bigl(1+2\sum_{j=1}^{i-1}\eps_j \Bigr)(1+\eps_i)\cdot\frac{\Delta}{2^i}\leq 
   \Bigl(1+2\sum_{j=1}^{i}\eps_j \Bigr)\cdot\frac{\Delta}{2^i}\leq
  \frac{\Delta}{2^{i-1}}.
  \]
  where here we make use of the inequality 
  $\sum_{j=1}^{i}\eps_j\leq\sum_{j=1}^t\eps_j\leq 1/2$.
  
We finish by applying Lemma~\ref{match-prop} on the hypergraph $H_t$ and vertex set $U_+$. Note that $H_t$ has maximum degree $O(\log n)$, and every node of $U_+$ has degree $\Theta(\log n)$. So, Lemma~\ref{match-prop} runs in time $O( r \log n \log r )$. The resulting matching $M$ contains at least $\Omega(|U_+|/r)$ vertices of $U_+$.  Since $\epsilon_i \geq \tfrac{1}{4} \log \Delta$, each application of \Cref{lemma:HG_splitting} takes time $O(r \log(n \Delta) \log^2 \Delta)$, and there are $t = O(\log \Delta)$ applications of it in total.
\end{proof}
  
By applying the hypergraph degree splitting of
Lemma~\ref{split2} repeatedly, we can now prove
\Cref{thm:hypergraph}.

{
\renewcommand{\thetheorem}{\ref{thm:hypergraph}}
\begin{theorem}
  Let $H$ be an $n$-node hypergraph with maximum degree at most
  $\Delta$ and rank at most $r$. Then, a maximal matching of $H$ can
  be computed in $O(r^2 \log (n\Delta) \log n \log^4 \Delta)$ rounds in the \LOCAL model.
\end{theorem}
\addtocounter{theorem}{-1}
}
\begin{proof}
If $\Delta \leq r \log^2 n$, then the line graph of $H$ has degree at most $r \Delta$, and so we simply use the deterministic MIS algorithm of \cite{BEK15} in time $O( r \Delta + \log^* (n \Delta)) \leq O(r^2 \log^2 n)$, and we are done. So let us assume that $\Delta \geq r \log^2 n$.

Each time we apply \Cref{split2}, we reduce the number of vertices in $U_+$ by a factor of $\Omega(1/r)$. Thus, after $O(r \log n)$ applications, we reduce the degree by a factor of $1/2$. So after $\log \Delta$ applications, the degree is reduced to $O(\log n)$. 

At this stage, we note that the residual line graph of $H$ has degree at most $O(r \log n)$. We can thus use a deterministic MIS algorithm on it, in time $O( r \log n + \log^*(n \Delta))$.

Each application of \Cref{split2} uses $O(r \log r \log n + r \log(n \Delta) \log^3 \Delta)$ time. Our assumption that $\Delta \geq r \log^2 n$ ensures
that the first term is dominated by the second term.
\end{proof}

\subsection{Implications on Edge-Coloring, Maximum Matching, and Low-Degree Orientation}
\begin{proof}[Proof of \Cref{crl:edge-coloring}] This follows immediately from \Cref{thm:hypergraph}, combined with a reduction of Fischer, Ghaffari, and Kuhn\cite{FischerGK17} that reduces $(2\Delta-1)$-list-edge-coloring of a graph $G = (V,E)$ to hypergraph maximal matching on hypergraphs of rank $3$, with $O(|V|+|E|)$ vertices and maximum degree $O(\Delta(G)^2)$, The complexity is thus $O(\log^2 n \log^4 \Delta)$.
\end{proof}

\begin{proof}[Proof of \Cref{crl:matching-approx}] The algorithm follows the approach of Hopcroft and Karp\cite{HopcroftKarp1973} based on repeated augmentation of the matching. Augmenting the matching $M$ with $P$ means replacing the matching edges in $P \cap M$ with the edges $P\setminus M$. For each $\ell =1$ to $2/\eps$, we find a maximal set of vertex-disjoint augmenting paths of length $\ell$, and we augment them all. Given a matching $M$, an augmenting path $P$ with respect to $M$ is a path that starts with an unmatched vertex, alternates between non-matching and matching edges, and finally ends in an unmatched vertex. Hopcroft and Karp \cite{HopcroftKarp1973} show that this produces a $(1+\eps)$-approximation of maximum matching. See also \cite{lotkerMatchingImproved}, which uses a similar method to obtain a $O(\log n/\eps^3)$-round randomized distributed algorithm for $(1+\eps)$-approximation of maximum matching, by applying Luby's algorithm~\cite{luby86}.

We now discuss how to compute a maximal set of vertex-disjoint augmenting paths of a given length $\ell$. Form the hypergraph $H$ on vertex set $V$, with a hyperedge $\{v_1, \dots, v_\ell \}$ for every augmenting path $v_1, \dots, v_{\ell}$. This hypergraph has rank at most $\ell$ and maximum degree at most $\Delta^{\ell}$; every maximal matching of $H$ corresponds to a length-$\ell$ vertex-disjoint augmenting path. Moreover, a single round of communication on $H$ can be simulated in $O(\ell)$ rounds of the base graph $G$.  Placing these parameters in the bound of \Cref{thm:hypergraph}, we get complexity $O( \ell \cdot \ell^2 \log(n\Delta^{\ell}) \log n \log^4(\Delta^{\ell}))=O(\log^2 n\log^5 \Delta/\eps^8)$ 
rounds for each value of $\ell$ up to $2/\eps$. Thus, the overall complexity is $O(\log^2 n\log^5 \Delta/\eps^9)$. 
\end{proof}

\begin{proof}[Proof of \Cref{crl:orientation}]
We follow the approach of Ghaffari and Su\cite{GS17}, which iteratively improves the orientation by reducing its maximum out-degree via another type of augmenting path. Let $D=\lceil\lambda (1+\eps)\rceil$. Given an arbitrary orientation, Ghaffari and Su call a path $P$ an augmenting path for this orientation if $P$ is a directed path that starts in a node with out-degree at least $D+1$ and ends in a node with out-degree at most $D-1$. Augmenting this path means reversing the direction of all of its edges. This would improve the orientation by decreasing the out-degree of one of the nodes whose out-degree is above the budget $D$, without creating a new node with out-degree above the threshold.

Let $G_0$ be the graph with our initial arbitrary orientation. Define $G'_{0}$ to be a directed graph obtained by adding a source node $s$ and a sink node $t$ to $G_0$. Then, we add $\text{outdeg}_{G_0}(u)-D$ edges from $s$ to every node $u$ with outdegree at least $D+1$, and $D-\text{outdeg}_{G_0}(u)$ edges from every node $u$ with outdegree at most $D-1$ to $t$. We improve the orientation gradually in $\ell=O(\log n/\eps)$ iterations. In the $i^{th}$ iteration, we find a maximal set of edge-disjoint augmenting paths of length $3+i$ from $s$ to $t$ in $G'_i$, and then we reverse all these augmenting paths. The resulting graph is called $G'_{i+1}$. Ghaffari and Su\cite[Lemma D.6]{GS17} showed that each time the length of the augmenting path increases by at least one, and at the end, no augmenting paths of length at most $\ell=O(\log n/\eps)$ remains. They used this to prove that there must be no node of out-degree $D+1$ left, at the end of the process, as any such node would imply the existence of an augmenting path of length at most $\ell$ \cite[Lemma D.9]{GS17}.

The only remaining algorithmic piece is to compute a maximal set of edge-disjoint augmenting paths of length at most $3+i < \ell$, in a given orientation. We do so using \Cref{thm:hypergraph}, by viewing each edge as one vertex of our hypergraph, and each augmenting path of length at most $3+i < \ell$ as one hyperedge. The round complexity is at most $O(\ell \cdot \ell^2 \log(n\Delta^{\ell}) \log n \log^4(\Delta^{\ell}))$ , where the first $\ell$ factor  is because simulating each hyperedge needs $\ell$ rounds. This is the complexity for each iteration. For $\ell=O(\log n/\eps)$ iterations, the total complexity is
$O(\log^{10} n\log^5 \Delta/\eps^9)$. 
\end{proof}

%% file: LLL.tex
\section{The Lov\'{a}sz Local Lemma}
\label{lll-sec1}
We will consider the following, somewhat restricted, case of the LLL: the probability space $\Omega$ is defined by variables $X(1), \dots, X(v)$; each $X(i)$ takes on values from some countable domain $D$. The variables $X(1), \dots, X(v)$ are all mutually independent. Every bad event $B \in \mathcal B$ is a boolean function of a set of variables $S_B \subseteq [v]$. We say that a configuration $X$ \emph{avoids $\mathcal B$} if every $B \in \mathcal B$ is false on $X$.

We define a \emph{dependency graph} $H$ for $\mathcal B$, to be an undirected graph on vertex set $\mathcal B$, with an edge $(A,B)$ if $S_A \cap S_B \neq \emptyset$; we write this as $A \sim B$, and we define $N(B) = \{A \mid A \sim B \}$; note that with definition we also have $B \sim B$. With this notation, we can state the LLL in its simplest ``symmetric'' form: if the dependency graph $H$ has maximum degree $d-1$, and every bad-event $B \in \mathcal B$ has $\Pr_{\Omega}(B) \leq p$, and $ e p d \leq 1$, then there is a positive probability that no $B \in \mathcal B$ occurs.

The LLL underlies a wide variety of combinatorial constructions. Thus, a distributed LLL algorithm is a key building-block for graph algorithms, such as frugal or defective vertex-colorings. In such settings, we have a communication graph $G$ and a variable for every vertex $x \in G$, and have a bad-event $B_x$ indicating that the coloring fails for $x$ in some way; for example, in a frugal coloring, a bad-event for $x$ may be that $x$ has too many neighbors with a given color. Each vertex $x$ typically only uses information about its $r$-hop neighborhood, where $r$ is very small (and typically is $O(1)$). In this case, the dependency graph $H$ is essentially the same as $G^r$. To avoid further confusion, we will assume that the communication graph is the same as the dependency graph; since communication rounds on $H$ may be simulated in $O(r)$ rounds of $G$, this typically only changes the overall runtime of our algorithms by a small (typically constant) factor. So we also assume that $| \mathcal B | = n$ and $d = \Delta+1$.

\subsection{Previous Distributed LLL Algorithms}
The LLL, in its original form, only shows that there is an exponentially small probability of avoiding the events of $\mathcal B$, so it does not directly give efficient algorithms. There has been a long history of developing algorithmic versions of the LLL, including distributed algorithms. A breakthrough result of Moser \& Tardos \cite{mt} gave one of the first general serial algorithms for the LLL; they also discussed a parallel variant, which can easily be converted into an distributed algorithm running in $O(\log^2 n)$ rounds. This algorithm converges under essentially the same conditions as the probabilistic LLL, namely it requires $ e p d (1 +\eps) \leq 1$ for some constant $\eps > 0$.

In \cite{pettie}, Chung, Pettie \& Su began to investigate the algorithmic LLL specifically in the context of distributed computations. They give an algorithm running in $O( \log^2 d \log n)$ rounds (subsequently improved to $O(\log d \log n)$ by \cite{gmis}).  They also discuss an alternate algorithm running in $O( \frac{ \log n}{ \log (e p d^2) } )$ rounds, under the stronger LLL criterion $ e p d^2 < 1$. They further showed that the criterion $ e p d^2 < 1$, although significantly weaker than the LLL itself, is still usable to construct a number of combinatorial objects, such as frugal colorings and defective colorings. We refer to this type of weakened LLL condition as a \emph{polynomially-weakened LLL criterion} (pLLL). We will use this algorithm as a key subroutine; note in particular that if we satisfy the pLLL criterion $p d^3 < 1$, then this algorithm runs in $O( \frac{\log n}{\log d})$ rounds.

More recently, Fischer \& Ghaffari \cite{ghaffari-lll} have algorithm running in $2^{O(\sqrt{ \log \log n})}$ rounds, under the pLLL criterion $p (e d)^{32} < 1$,  as long as $d < (\log \log n)^{1/5}$. Although the general LLL algorithm of \cite{ghaffari-lll} has a significant limitation on degree, they nevertheless used it to construct a number of graph colorings including defective coloring, frugal coloring, and vertex coloring in $2^{O(\sqrt{\log \log n})}$ time (for arbitrary degree graphs).  On the other hand, \cite{brandt} has shown a $\Omega (\log \log n)$ lower bound on the round complexity of distributed LLL, even under a pLLL criterion.

In addition to these general algorithms, there have been a number of algorithmic approaches to more specialized LLL instances. In \cite{chang2018complexity}, Chang et al. have investigated the LLL when the graph $G$ is a tree; they develop an $O(\log \log n)$ round algorithm in that case. In \cite{harris-llll}, Harris developed an $O(\log^3 n)$-round algorithm for a form of the LLL known as the Lopsided Lov\'{a}sz Local Lemma (which applies to more general types of probability spaces). Finally, there have been a number of parallel PRAM algorithms developed for the LLL, including \cite{mt, harris2}; these are often similar to distributed LLL algorithms but are not directly comparable.

\subsection{Graph Shattering}
The LLL algorithms use a general technique for building distributed graph algorithms known as \emph{graph shattering.} These algorithms have two phases. In the first
phase, there is some random process, which satisfies most of the vertices. The choices made by these ``good'' vertices are permanently commited to, leaving us with a residual set $R$ of unsatisfied vertices. The connected components of $G[R]$ are small with high probability. In the second phase, a deterministic algorithm is used to solve the residual problem on each
component of $G[R]$. We provide a technical overview of this process in Appendix~\ref{shattering-appendix}, which we summarize as follows here:
\begin{theorem}
  \label{shattering-thm}
  Suppose each vertex $v$ survives to a residual graph $R$ with
  probability at most $(e \Delta)^{-4 c}$, and this bound holds even
  for adversarial choices of the random bits outside the $c$-hop
  neighborhood of $v$ for some constant $c \geq 1$.

  Then w.h.p each connected component of the residual graph has size at most $O(\Delta^{2 c} \log n)$.
  
  If the residual problem can be solved via a $\seqloc[r]$
  procedure, then the residual problem can be solved w.h.p.\ in the \LOCAL
  model in $\min \big \{ r 2^{O(\sqrt{\log \log n})}, O( r(\Delta(G^r) + \log^* n) ) \big \}$ rounds.
  \end{theorem}

\subsection{Bootstrapping}
There is another important subroutine used in our LLL algorithm, referred to as \emph{bootstrapping}. This is a technique wherein LLL algorithms can be used to ``partially'' solve a problem, generating a new LLL problem instance whose slack is ``amplified.'' This technique was first introduced by \cite{chang2017time}, and extended to more general parameters in \cite{ghaffari-lll}. 

When developing our LLL algorithms, there are two ways of measuring success. First, there is the global guarantee: what is the probability that all $\mathcal B$ is avoided? Second, there is a local guarantee: for any fixed bad-event $B$, what is the probability that the variable assignment generated by our algorithm makes $B$ be true? We say that the algorithm has \emph{local failure probability} $\rho$ if every $B \in \mathcal B$ has a probability at most $\rho$ of being false. This may be much smaller than the global failure probability.

Bootstrapping is a method of converting an LLL algorithm to solve $\mathcal B$, which has some guaranteed local failure probability, into a new LLL instance which represents $\mathcal B$. We can summarize it as follows:
\begin{lemma}
  \label{boot-lemma}
  Let $A$ be a randomized LLL algorithm with runtime $r$ and local failure probability $\rho$ to solve an LLL instance $\mathcal B$ with parameters $p, d$. Then we can generate a new LLL instance $\mathcal C$ on a communications/dependency graph $H$ such that:
  \begin{enumerate}
  \item $|\mathcal C| = |\mathcal B|$.
  \item One communication round of $H$ can be simulated in $O(r)$ rounds of $G$.
  \item If we solve $\mathcal C$ then we can generate a solution for $\mathcal B$ in $O(r)$ rounds of $G$. 
  \item $\mathcal C$ has parameters $p' = \rho$ and $d' = d^{2 r}$
  \end{enumerate}
\end{lemma}
\begin{proof}
  When we run $A$, this generates some value for the underlying variables $X = X(1), \dots, X(v)$. Consider some event $B \in \mathcal B$, and let us define a bad-event $C_B$ to be that $B$ is false on $X$ after the termination of $A$. The event $C_B$ can be thought of as a boolean function of the random variables generated during the execution of $A$.

We claim now that if we can solve the LLL instance $\mathcal C = \{ C_B \mid B \in \mathcal B \}$, then we automatically solve $\mathcal B$ as well. For, if no event $C_B$ holds, then all the events $B \in \mathcal B$ are false at the termination of $A$, and therefore the variables $X$ generated by that algorithm solve the underlying LLL instance $\mathcal B$.

By definition, each bad-event $C_B$ has probability at most $p' = \rho$. Since the \LOCAL algorithm $A$ runs for $r$ steps, the events $C_B$ and $C_{B'}$ can only affect each other if $\text{dist}_G(B, B') \leq 2 r$. Therefore, each event $C_B$ can only affect at most $d' = \Delta^{2 r}$ other bad-events in $\mathcal C$.
\end{proof}

We typically use Lemma~\ref{boot-lemma} to combine two LLL algorithms $A_1, A_2$ into a new hybrid algorithm $A$: we run $A_1$ on our original problem instance $\mathcal B$, generating a new LLL instance $\mathcal C$ which is then solved by $A_2$.

The local success probability of an algorithm is closely linked to the network size parameter $n$. It is possible to run a \text{LOCAL} algorithm $A$ with an alternate choice of $n$; we say in this case that we \emph{bootstrap} $A$ with parameter $n'$ and write the resulting algorithm as $A^{[n']}$. This includes generating new ID's for all the vertices, which will be random bit-strings of length $10 \log n'$. This is often referred to as running $A$ with a ``fake'' value of $n$. In this situation, we can no longer guarantee that $A$ has a high probability of \emph{globally} solving the problem; for one thing, it is possible that two nodes will select the same vertex ID. However, if $A$ runs in time $r(n)$ for a slowly-growing function $r$, we can often guarantee a high \emph{local} success probability.

\begin{proposition}
  \label{bootstrap-prop}
 If $\Delta^{r(w)} \leq w$, then $A^{[w]}$ has local failure probability at most $2/w$.
\end{proposition}
\begin{proof}
  Let $B \in \mathcal B$. Since $A^{[w]}$ is a local algorithm with radius $r(w)$, the behavior of $B$ when run on the graph $G$ will be the same as when run on the graph $G[U]$ where $U$ denotes the set of all vertices within distance $r(w)$ of $B$. We can bound $|U| \leq \Delta^{r(w)} \leq w$.

  The probability that two nodes of $U$ select the same ID is at most $\binom{|U|}{2} 2^{-10 \log w} \leq w^{-4}$. If all vertices have distinct ID's then by definition of $A$ there is a probability of at least $1 - 1/w$ that $A^{[w]}$ succeeds on the \emph{entire graph} $G[U]$; in particular, it succeeds on $B$ with probability at least $1 - 1/w$. Overall, the total success probability is at least $1 - 2/w$.
\end{proof}

\subsection{The LLL for Low-Degree Graphs}
As we have discussed, the algorithm of \cite{ghaffari-lll} is very fast, running in just $2^{O(\sqrt{\log \log n})}$ rounds, but only works for graphs whose degree is very low. In this section, we will use one of the core subroutines of \cite{ghaffari-lll} to obtain an algorithm running in $2^{O(\sqrt{\log \log n})}$ rounds for much larger degree; specifically, we will allow the degree to become as large as $d = 2^{\sqrt{\log \log n}}$, an exponential improvement over \cite{ghaffari-lll}. 

In analyzing this and other LLL algorithms, it is convenient to extend the domain $D$ by adding an additional symbol denoted \q; we say $X(i) = \q$ to indicate that variable $X(i)$ is not determined, but will be later drawn from $D$ with its original sampling probability. We let $\overline D = D \cup \{ \q \}$.

Given any vector $x \in \overline D^v$, and an event $E$ on the space $\Omega$, we define the \emph{marginal probability} of $E$ with respect to $x$ as the probability that $E$ holds, if all variables $i$ with $X(i) = \q$ are resampled from the original distribution. Note that if $x \in D^v$ then the marginal probability of any event with respect to $x$ is either zero or one. Also, if $x = (\q, \dots, \q)$, then the marginal probability of $E$ with respect to $x$ is simply $\Pr_{\Omega}(E)$.

For certain algorithms, it is convenient to assume that $v < \infty, D = \{0, 1 \}$ and $\Omega$  sets $\Pr(X(i) = 0) = \Pr(X(i) = 1) = 1/2$; we refer to this by saying that $\Omega$ is in \emph{normal form.} Any probability distribution $\Omega$ with a finite support can be discretized into a probability distribution $\Omega'$ in normal form, with an arbitrarily small increase in the size of $p$. The resulting communications graph remains the same. When converting a probability space into normal form, the number of variables may increase exponentially (or even more). For the algorithms we encounter in Section~\ref{lll-sec1}, which do not depend on $v$, this is harmless. In Section~\ref{lll-sec2}, we will see some other LLL algorithms which depend on the variable in a more crucial way, and in those cases we cannot afford to transform $\Omega$ into a normal form.

We summarize the main subroutine of \cite{ghaffari-lll}, which is inspired by sequential LLL algorithms of Molloy \& Reed \cite{molloy-reed} and Pach \& Tardos \cite{pach2009conflict}.
\begin{algorithm}[H]
\centering
\caption{Distributed LLL algorithm}
\label{distalg1a}
\begin{algorithmic}[1]
\State Initialize $K \leftarrow \emptyset$; this will be the set of \emph{frozen variables.}
\State Initialize $X = (\q, \dots, \q)$.
\State Compute a $d^2+1$-coloring $\chi$ of $G^2$.
\For{ $i = 1, \dots, d^2 + 1$}:
\For{each bad-event $B$ with $\chi(B) = i$}
\For{each $j \in S_B$}
\If{$j \notin K$ and $X(j) = \q$}
\State Draw $X(j)$ from its distribution under $\Omega$
\State\textbf{if} any $A \sim B$ has marginal probability at least $(e d)^8 p$ under $X$ \textbf{then} update $K \leftarrow K \cup S_A$
\EndIf
\EndFor
\EndFor
\EndFor
\end{algorithmic}
\end{algorithm}

Let $X' \in \overline D$ be the final value of vector $X$, and let $\mathcal A \subseteq \mathcal B$ be the set of all bad-events whose marginal probability under $X'$ is non-zero. We use the following key facts about Algorithm~\ref{distalg1a}:
\begin{theorem}[\cite{ghaffari-lll}]
\label{prev-lll-thm}
 Algorithm~\ref{distalg1a} runs in $O(d^2 + \log^*n)$ rounds. If $\Omega$ is in normal form, then at its termination:
\begin{enumerate}
\item Every $B \in \mathcal B$ has marginal probability under $X$ of at most $2 (e d)^8 p $.
\item Every $B \in \mathcal B$ has $\Pr(B \in \mathcal A) \leq (e d)^{-8}$; furthermore, this probability bound holds even if the random bits made outside a two-hop radius of $B$ are chosen adversarially.
\end{enumerate}
\end{theorem}

\begin{lemma}
\label{prev-lll-thm2}
If $20000 d^8 p \leq 1$, then a configuration avoiding $\mathcal B$ can be found in  $O(d^2) + 2^{O(\sqrt{\log \log n})}$ rounds w.h.p.
\end{lemma}
\begin{proof}
  If $d \geq \log n$, then the Moser-Tardos algorithm itself runs in $O(\log^2 n) = O(d^2)$ time. So we assume that $d < \log n$. We also assume wlg that $\Omega$ is in normal form.
  
Run Algorithm~\ref{distalg1a} in $O(d^2 + \log^*n)$ rounds; let $X' \in \overline{D}$ denote the partial solution that it generates. We can view the residual problem, of converting the partial solution $X'$ into a full solution $X$, as an LLL instance on the bad-events $\mathcal A$.  By Theorem~\ref{prev-lll-thm},  each $B \in \mathcal A$ has marginal probability at most $q= 2 (e d)^{8} p$. Our condition $20000 d^{8} p \leq 1$ ensures that $e d q \leq 1$, and so $\mathcal A$ satisfies the symmetric LLL criterion with slack $\eps = 1/2$. Therefore, the algorithm of \cite{mt} gives a Las Vegas procedure to solve  this residual problem.

By Theorem~\ref{prev-lll-thm}, each $B \in \mathcal B$ survives to the residual with probability $(e d)^{-8}$. By Theorem~\ref{shattering-thm} with $c = 2$, each connected component in the residual has size at most $N = O(\Delta^{2 c} \log n)$. The algorithm of \cite{mt} gives a Las Vegas algorithm on each component in $O(\log^2 N) = O( (\log \log n)^2 )$ rounds (here we use our bound $d < \log n$). By Proposition~\ref{prop:ZtoS}, this also yields an $\seqloc[O( (\log \log n)^2 )]$ algorithm on each component. By Theorem~\ref{shattering-thm}, therefore, the residual problem can be solved in $2^{O(\sqrt{\log \log n})}$ rounds.
\end{proof}

For small values of $d$, there is another way to run the above algorithm, which will be critical in many of our bootstrapping constructions.
\begin{lemma}
\label{base-thm-2}
Suppose $20000 d^{10} p \leq 1$ and $d \geq \log^* n$. Then there is randomized \LOCAL algorithm running in $O(d^2)$ rounds and with local failure probability at most $e^{-\frac{1}{10000 d^{10} p}}$.
\end{lemma}
\begin{proof}
  Let $\lambda = \frac{1}{20000 d^{10} p} \geq 1$.  Assume wlg with $\Omega$ is in normal form. We run Algorithm~\ref{distalg1a} to obtain a partial solution $X' \in \overline{D}$. Now, consider the residual problem of converting the partial solution $X'$ into a full solution $X$.  By Theorem~\ref{prev-lll-thm},  each $B \in \mathcal A$ has marginal probability at most $q = 2 (e d)^{8} p$, and this depends only on the two-hop neighborhood of $B$. 

Let $\mathcal R$ be the connected component containing $B$ in the residual graph and let $N = |\mathcal R|$ (if $B$ is false, then $\mathcal R = \emptyset$). There is a trivial $1$-round randomized algorithm for the residual problem, based on selecting a random draw from $\Omega$ and checking if $B$ is true. The overall probability that this randomized algorithm fails is at most $N q$. Therefore, if $N q < 1$, this gives a Las Vegas algorithm for the component of $B$. By Proposition~\ref{prop:ZtoL} with $r = 1$, the component can be solved deterministically in $O(d^2 + \log^* n)$ rounds.

Now suppose we run the above deterministic algorithm, without knowing the precise size of the component of $B$. As long as $N < 1/q$, this will succeed, and cause $B$ to be false. Therefore, the only way in which $B$ can be true, is if its component has size at least $1/q$. By Theorem~\ref{con-size0}, we estimate this as:
$$
\Pr(N \geq 1/q) \leq (e d)^{-\frac{1}{q d^2} + 1}  \leq (e d)^{1-\frac{1}{2 e^8 d^{10} p}} \leq (e d)^{ 1 - 3 \lambda}
$$

Since $\lambda \geq 1$, this in turn is at most $e^{-2 \lambda}$ as desired. So we have described an algorithm $A_1$ running in time $O (d^2 + \log^* n)$. We can remove the factor of $\log^* n$ by our assumption on the size of $d$.
\end{proof}

\subsection{Bootstrapping the Two Algorithms}
In order to get faster runtimes, we will combine the algorithms of Lemmas~\ref{prev-lll-thm2} and~\ref{base-thm-2}, via a series of bootstrapping steps.

\begin{proposition}
  \label{base-thm-3}
  Suppose $d^{15} p \leq 1$ and $d \geq \log^* n$. Then for $n$ sufficiently large and any $i \geq 0$, one can transform $\mathcal B$ into a new problem instance $\mathcal C$ with parameters
  $$
  p' = \frac{1}{\exp^{(i)}(\frac{1}{d^{15} p})},  \qquad d' = \exp^{(i)}(d^3)
  $$

The communication graph for $\mathcal C$ can be simulated in $\exp^{(i-1)} (d^3)$ rounds with respect to $G$.
\end{proposition}
\begin{proof}
  Let us define $\lambda = \frac{1}{d^{15} p}$. We will use Lemma~\ref{boot-lemma} recursively, generating a series of problems $\mathcal C_1, \mathcal C_2, \dots, \mathcal C_i$ such that each $\mathcal C_i$ has probability $p_i$, dependency $d_i$, and can be simulated in $G$ in $r_i$ rounds.

  To begin, we set $\mathcal C_0 = \mathcal B$, giving parameters $p_0 = p, d_0 = d, r_0 = 1$. For the recursive step, we apply Lemma~\ref{boot-lemma} to Lemma~\ref{base-thm-2}. As long as $20000 d_i^{10} p_i \leq 1$, this generates a new problem $\mathcal C_{i+1}$ with parameters
$$
  p_{i+1} = e^{-2 \lambda_i}, \qquad  d_{i+1} = d_i^{2 a d_i^2}, \qquad r_{i+1} = (a d_i)^2 \times r_i
  $$
  where $\lambda_i = \frac{1}{20000 d_i^{10} p_i} \geq 1$ and where $a > 0$ is some universal constant. To explain the recursion for $r_i$, observe that each round of $\mathcal C_{i+1}$ can be simulated in $O(d_i^2)$ with respect to the communications graph of $\mathcal C_{i}$, each round of which in turn can be simulated in $r_i$ rounds with respect to $\mathcal C_i$. 

  Our first task is to show that $\lambda_i \geq 1$ for all $i \geq 1$, i.e. that this recursion never gets stuck. In fact, we will show a stronger claim, that $\lambda_i \geq d_i^4$ for all $i$. We show this by induction on $i$. The base case $i = 0$ is guaranteed by our hypothesis and $d \geq \log^* n \gg 1$. For the induction, we have:
  \begin{align*}
    \lambda_{i+1}/d_{i+1}^4 &= \frac{1}{20000 d_{i+1}^{14} p_{i+1}} = \frac{e^{2 \lambda_i}}{20000 d_i^{28 a d_i^2}} = \frac{e^{2 \lambda_i}}{20000 e^{28 a d_i^2 \log d_i}} \geq e^{2 \lambda_i - 10 - 28 a d_i^2 \log d_i}
  \end{align*}

  By our induction hypothesis, $\lambda_i$ is larger than $d_i^4$, which (for $d_i \geq v_0$ and $v_0$ a sufficiently large constant) is greater than $10 + 28 a d_i^2 \log d_i$. Thus, this is at least $e^{\lambda_i} \geq 1$.

  We have therefore shown $\lambda_{i+1}/d_{i+1}^4 \geq e^{\lambda_i} \geq 1$. So the recurrence does not terminate. Consequently we also have $\lambda_i \geq \exp^{(i)}(\lambda)$. This gives:
  $$
  p_i = \frac{1}{20000 d_i^{10} \lambda_i} \leq \frac{1}{\lambda_i} \leq \frac{1}{\exp^{(i)}(\lambda)}
  $$

  Next, let us calculate $d_i$. We show by induction on $i$ that
  $$
  d_i \leq \exp^{(i)} (d^{2.5} - 2^{-i+3} )
  $$

  The base cases are $i = 1$ and $i = 2$, and these bounds can easily been shown to hold for $d$ sufficiently large. For the induction $i \geq 2$ and $d$ sufficiently large, we get:
  \begin{align*}
    d_{i+1} &= \exp(2 a d_i^2 \log d_i) \leq e^{d_i^{2.1}} = \exp^{(2)} (2.1 \log d_i)  \leq \exp^{(3)} (0.75 +  \exp^{(i-2)} (d^{2.5} - 2^{-i+3}))
    \end{align*}
    
A simple induction shows on $j$ shows that for all $j \geq 0$ and $a,b \geq 0$ it holds that $\exp^{(j)}(a) + b \leq \exp^{(j)} (a + b 2^{-j})$.   Thus, with $j = i-2, a = 0.75$ and $b = \exp^{(i-2)}(d^{2.5} - 2^{-i+3})$, we can calculate
    $$
   d_{i+1} \leq \exp^{(3)} ( \exp^{(i-1)} (d^{2.5} - 2^{-i+2} + 2^{-i+2})) =\exp^{(i+2)} (d^{2.5} - 2^{-i+1})
   $$

  Our bound on $d_i$ has now been established. Finally, for our bound on $r_i$, we get
  $$
  r_i  \leq \prod_{j=0}^{i-1} (a d_j) \leq \prod_{j=0}^{i-1} a \exp^{(j)} (d^{2.5} - 2^{-j+3} )
  $$
  Routine calculations show that this is at most $\exp^{(i-1)} (d^3)$.
\end{proof}

Through the magic of bootstrapping, these two LLL algorithms can be combined to give faster algorithms for small values of $d$.
\begin{theorem}
\label{hthm1}
Let $i$ be an integer in the range $1 \leq i \leq \log^{*} n - 2 \log^* \log^* n$. If $20000 d^{8} p \leq 1$, then there is an algorithm to find a configuration avoiding $\mathcal B$ w.h.p. in time $\exp^{(i)} \bigl( C (\log d + \sqrt{\log^{(i+1)} n}) \bigr)$, for some universal constant $C$.
\end{theorem}
\begin{proof}
  Let us first show that this holds when $d \leq e^{\sqrt{\log^{(i+1)} n}}$ and $n$ is larger than any needed constants. All the asymptotic notations in this proof refer to universal constants (which do not depend upon $i$).

  Let $A$ be the algorithm of Lemma~\ref{prev-lll-thm2}. We begin by bootstrapping $A^{[w]}$ for $w = (\log^{(i-1)} n)^{20}$. This runs in time $r = e^{O(\log d + \sqrt{\log \log w})} \leq e^{O(\sqrt{\log^{(i+1)} n})}$.  Note that our condition on $i$ ensures that $\log^{(i)} n \geq \log^* n$ and so $w \rightarrow \infty$ as $n \rightarrow \infty$. So, our assumption on $n$ ensures that $w$ will be larger than any needed constants.

We want to apply Lemma~\ref{boot-lemma}. Since we need to get a bound on the local success probability, Proposition~\ref{bootstrap-prop} requires us to first show that $d^{r} \leq w$. For this, we have:
\begin{align*}
  d^{r} &\leq \exp(\sqrt{\log^{(i+1)} n} \times e^{O(\sqrt{\log^{(i+1)} n})}) \leq \exp^{(2)}(O(\sqrt{\log^{(i+1)} n})) \leq \exp^{(2)}(O(\sqrt{\log \log w}))
\end{align*}
and this will be smaller than $w$ for $w \geq C_0$ and $C_0$ a sufficiently large constant, which holds as $n \rightarrow \infty$.

Therefore, Lemma~\ref{boot-lemma} generates a new problem instance $\mathcal C$ with $p' = 2/w = 2 (\log^{(i-1)} n)^{-20}$ and $d' = d^{2 r} = \exp^{(2)} \bigl(O(\sqrt{\log^{(i+1)} n}) \bigr)$. 

We will use Lemma~\ref{base-thm-3} to further amplify $\mathcal C$. The condition that $i \leq \log^* n - 2  \log^* \log^* n$ implies that $d' \geq \log^* n$ for $n$ sufficiently large. We compute the parameter $(d')^{15} p'$ for $n$ sufficiently large as
\begin{align*}
(d')^{15} (p')  = \frac{2 \exp^{(2)} \bigl(O(\sqrt{\log^{(i+1)} n}) \bigr)}{(\log^{(i-1)} n)^{20}} = \frac{\exp^{(2)} \bigl(O(\sqrt{\log \log w}) \bigr)}{w^{20}} \leq w^{-19}
\end{align*}

Accordingly, Proposition~\ref{base-thm-3} generates a new problem instance $\mathcal C_{i-1}$, with parameter
$$
p'' = 1/\exp^{(i-1)} (\frac{1}{(d')^{15} p'}) \leq \frac{1}{\exp^{(i-1)}( (\log^{(i-1)} n)^{19})}.
$$
and this is at most $n^{-20}$ for $n$ sufficiently large.

Since the local failure probability of $\mathcal C_{i-1}$ is so small, we can find a satisfying assignment w.h.p. by sampling $\mathcal C_{i-1}$ a single time. So, the run-time for our algorithm is simply the cost of simulating a single round of $\mathcal C_{i-1}$, i.e. it is $\exp^{(i-2)} O(d'^3) = \exp^{(i)}  \bigl(O(\sqrt{\log^{(i+1)} n}) \bigr)$.

Finally, let us discuss the case that $d \geq e^{\sqrt{\log^{(i+1)} n}}$. In this case, set $s$ such that that $d = e^{\sqrt{\log^{(i+1)} s}}$. We run the above algorithm $A$ with a fake parameter $n' = s \geq n$. This gives a runtime of $\exp^{(i)} \bigl(O(\sqrt{\log^{(i+1)} s}) \bigr)$, which is equal to $\exp^{(i)}(O(\log d))$ as desired.
\end{proof}

\textbf{Explanation of the LLL algorithm runtime.} We have shown an algorithm with a runtime of
\begin{equation}
\label{rttr1}
\min_{i \in \{1, \dots, \log^*n - 2 \log^* \log^* n \}} \exp^{(i)} (C ( \log d + \sqrt{\log^{(i+1)} n})).
\end{equation}

Let us first remark that it is possible to show Theorem~\ref{hthm1} still holds under the weaker constraint $i \leq \log^*n - O(1)$. To summarize, in \Cref{base-thm-3} we would precompute a coloring of the graph $G^{r_{i-1}}$; this can be done in $O(r_{i-1} (d_{i-1} + \log^* n))$ rounds; we can use this coloring for the remainder of the algorithm to avoid all the $\log^* n$ coloring steps. We have stated the result with the stronger restriction on $i$ in order to simplify the proofs throughout the paper.

To put it mildly, the function (\ref{rttr1}) is not a familiar function of $n$ and $d$. We aim to provide some intuition as to its asymptotic behavior. Ignoring integrality constraints, and using the fact that $i$ might get as large as $\log^* n - O(1)$, we may set $i = \log^*n - \log^* K$ for a constant $K$ to get:
$$
  \log n = \exp^{(i)} ( \log K ), \qquad  \log \log n = \exp^{(i)} ( \log \log K )
  $$

For constant $d$ and using this value of $i$, we see that our LLL algorithm has a runtime of
$$
\exp^{(i)} ( C \sqrt{\log K} )
$$

For large enough $K$, the value of $C \sqrt{\log K}$ is a constant which lies strictly between the constants $\log K$ and $\log \log K$. Thus, our algorithm runtimes lies somewhere in the window between $\log n$ and $\log \log n$. Here is an alternate way of looking at this: we can consider a class of runtimes of the form $\log^{(j)} n$, where $j = 1$ corresponds to $\log K$ and $j = 2$ corresponds to $\log \log K$. In this sense, the value $C \sqrt{\log K}$ corresponds to some real number $j \in (1,2)$. In other words, for constant $d$ our LLL algorithm runs in time $\log^{(1+\eps)} n$ for some \emph{constant} value $\eps > 0$. Of course, this notation is not really well-defined. But, it indicates that we are making a constant amount of progress from a $\log n$ runtime to $\log \log n$ runtime.

\subsection{Derandomization}
We can use our generic derandomization framework to obtain a deterministic LLL algorithm. The main difference is that the running times all have their dependence on $n$ stepped up by an exponential factor.
\begin{theorem}
  \label{derand-lll-thm}
  Let $i$ be an integer in the range $1 \leq i \leq \log^{*} n - 2 \log^* \log^* n$. If $20000 d^{8} p \leq 1$, then there is a deterministic \LOCAL algorithm to find a configuration avoiding $\mathcal B$ in time $\exp^{(i)} \bigl(C (\log d + \sqrt{\log^{(i)} n}) \bigr)$, for some universal constant $C$.
\end{theorem}
\begin{proof}
  Let us first consider $i = 1$; in this case, we want to show that the algorithm runs in $\poly(d) \times 2^{O(\sqrt{\log n})}$ rounds.   Lemma~\ref{prev-lll-thm2} gives a randomized algorithm running in $r = O(d^2) + 2^{O(\sqrt{\log \log n})}$ rounds. This problem is locally-checkable, so this is also a $\zloc[r+1]$ algorithm. Therefore, Proposition~\ref{prop:ZtoL} gives a deterministic algorithm to avoid $\mathcal B$ in $r 2^{O(\sqrt{\log n})}$ rounds, which in this case is at most $d^2 \times 2^{O(\sqrt{\log n})}$.

  Next, suppose $i > 1$. In this case, applying Theorem~\ref{hthm1} with parameter $i - 1$ yields a Las Vegas algorithm running in $r = \exp^{(i-1)} \bigl(O(\log d + \sqrt{\log^{(i)} n}) \bigr)$ rounds. Proposition~\ref{prop:ZtoL} transforms this to a deterministic algorithm to avoid $\mathcal B$ in $d^{O(r)} + O(\log^* n)$ rounds. Simple analysis shows that this is at most $e^{O(r \log d)}$ (the condition that $i \leq \log^*n - 2 \log^* \log^* n$ ensures that the $\log^* n$ term is negligible). Furthermore, in the expression $e^{O(r \log d)}$, the $\log d$ term can be absorbed into the constant term of $\exp^{(i-1)} (O(\log d))$ in the expression for $r$, so this is at most $\exp^{(i)} \bigl(O(\log d + \sqrt{\log^{(i)} n}) \bigr)$ as claimed.
  \end{proof}

\section{The LLL for High-Degree Graphs}
\label{lll-sec2}
There are two main shortcomings with the distributed LLL algorithm of Section~\ref{lll-sec1}. First, it is only a pLLL criterion; second, it requires $\Delta \leq 2^{O(\sqrt{\log \log n})}$.   In this section, we describe an alternative LLL algorithm, whose run-time does not depend on $\Delta$. We use this algorithm to provide distributed algorithms for problems such as defective and frugal vertex coloring, running in time $2^{O(\sqrt{\log \log n})}$. 

This matches the run-time of \cite{ghaffari-lll} for these problems. However, the key difference is that the latter work uses ad-hoc and problem-specific techniques; for example, in defective vertex coloring, certain vertices are re-colored using a secondary set of colors, instead of resampling a bad-event directly. This negates to a large extent one of the main advantages to developing general LLL algorithms. The algoirithm here can be described in a much more generic, high-level way. In addition, we provide algorithms for other LLL problems, such as $k$-SAT with bounded clause overlap, that do not appear directly possible in the framework of Lemma~\ref{prev-lll-thm2} or the algorithm of \cite{ghaffari-lll}.

Our algorithm is quite similar to an LLL algorithm of \cite{pettie}; the key definition underlying the algorithms is that of a \emph{dangerous} event.

\begin{definition}[Dangerous event]
Let $B \in \mathcal B$ and $x \in D^v$. For any $U \subseteq N(B)$, define  $y_U \in \overline D^v$ by 
$$
y_U(i) = \begin{cases}
\q & \text{if $i \in \bigcup_{A \in U} S_A$} \\
x(i) & \text{otherwise}
\end{cases}
$$

We refer to $y_U$ as the \emph{reversion of $x$ with respect to $U$.} We say that $B$ is \emph{$q$-dangerous with respect to $x$}, if there is any $U \subseteq N(B)$ such that the marginal probability of $B$ with respect to $y_U$ is at least $q$.
\end{definition}

Intuitively, a dangerous event is one that could have high probability if we revert some of its neighbors. With this definition, we can provide a sketch of the LLL algorithm:
\begin{algorithm}[H]
\centering
\caption{Distributed LLL algorithm}
\label{distalg2}
\begin{algorithmic}[1]
\State Draw $X \in D^v$ from the distribution $\Omega$
\State Construct the set $M \subseteq \mathcal B$ consisting of all bad-events which are $(e d)^{-3}$-dangerous with respect to $X$.
\State Form the vector $Y \in \overline D^v$ by
$$
Y(i) = \begin{cases}
\q & \text{if $i \in \bigcup_{B \in M} S_B$} \\
X(i) & \text{otherwise}
\end{cases}
$$
\State Form the set $\mathcal A \subseteq \mathcal B$ by 
$
\mathcal A = \{ B \in \mathcal B \mid N(B) \cap M \neq \emptyset  \}
$
\State Fix all the non-$\q$ entries of $Y$, and use the deterministic LLL algorithm on (each connected component) of $\mathcal A$ to fill in all the $\q$ entries of $Y$.
\end{algorithmic}
\end{algorithm}

The key to analyzing Algorithm~\ref{distalg2}, as in \cite{pettie}, will be to bound the probability that an event is $q$-dangerous. In \cite{pettie}, the following bound was provided:
\begin{proposition}[\cite{pettie}]
Any $B \in \mathcal B$ is $q$-dangerous with probability at most $2^d p / q$.
\end{proposition}

Unfortunately, this bound is exponential in $d$, and so this typically leads to criteria  which are much weaker than the LLL (i.e., bounds which are exponential in $d$ as opposed to polynomial in $d$.) We are not aware of any stronger bound, \emph{as a function solely of $p$ and $d$}. However, in many problem instances, we can use the specific form of the bad-events $\mathcal B$ to give much stronger bounds on the probability of being $q$-dangerous. Our LLL criterion will thus use more information than just $p$ and $d$. Nevertheless, the LLL criterion will be \emph{local}, in the sense that it uses only information directly affecting any given bad-event (and not global information such as $n$). Furthermore, this bound can be often be computed fairly readily, and once it is computed no further information about $\mathcal B$ will be used.

We will define a statistic we refer to as \emph{fragility} and use this to compute the probability of a bad-event becoming dangerous.

\begin{definition}
\label{fragile-definition}
Let $B$ be an event on variables $X(i), \dots, X(v)$ and and $X_0, X_1 \in D^v$. For any vector $a \in \{0, 1 \}^v$, define a configuration $Z_a \in D^v$ by $Z_a(i) = X_{a(i)} (i)$, and define the event $E_B$ by
$$
E_B = \bigvee_{a \in \{0, 1 \}^v} \text{$B$ holds on configuration $Z_a$}
$$

The \emph{fragility of $B$}, denoted $f(B)$, is defined to be the probability of $E_B$, when $X_0, X_1$ are drawn independently from $\Omega$.
\end{definition}

\begin{proposition}
\label{dangprop1}
For any bad-event $B$ and any $q \in [0,1]$ we have $\Pr(\text{$B$ is $q$-dangerous}) \leq f(B)/q$.
\end{proposition}
\begin{proof}
Consider drawing $X_0$ according to the distribution $\Omega$. Suppose there is some subset $Y \subseteq N(B)$ such that the reversion of $X_0$ with respect to $Y$ causes the marginal probability of $B$ to exceed $q$. Define the vector $a$ by
$$
a(i) = \begin{cases}
1 & \text{if $i \in \bigcup_{A \in Y} S_A$} \\
0 & \text{otherwise}
\end{cases}
$$
By definition of $q$-dangerous, when we draw $X_1$ according to the distribution $\Omega$, the probability that $B$ holds on $Z_a$ is greater than $q$. Therefore, the event $E_B$ has probability at least $\Pr(\text{$B$ is $q$-dangerous}) \times q$.
\end{proof}

\begin{proposition}
\label{u1}
Let $F = \max_{B \in \mathcal B} f(B)$. Any bad-event $A$ goes into $\mathcal A$ with probability at most $e F d^4$, and this depends only on the random bits within the two-hop neighborhood of $A$.
\end{proposition}
\begin{proof}
Each $B$ goes into $M$ if it is $q$-dangerous with respect to $X$, where $q = (e d)^{-3}$; by Proposition~\ref{dangprop1} this has probability at most $f(B)/q \leq F/q$. Also, the event that $B$ goes into $M$ can be determined solely by the values of $X(i)$ for $i \in S_B$. Thus, the event that $B$ goes into $M$ is not affected by random choices made outside the neighborhood of $B$.

The event $A$ goes into $\mathcal A$ only if there is some $B \sim A$ with $B \in M$. Taking a union bound over all $B \sim A$, we see that $\Pr(A \in \mathcal A) \leq \sum_{B \in N(A)} \Pr(B \in M) \leq d F/q$. Also, this only depends on random variables within the two-hop neighborhood of $A$.
\end{proof}

{
\renewcommand{\thetheorem}{\ref{tr1}}
\begin{theorem}
Suppose that $f(B) \leq e^{-10} d^{-12}$ for every $B \in \mathcal B$. Then \Cref{distalg2} terminates with a satisfying assignment in $2^{O(\sqrt{\log \log n})}$ rounds w.h.p.
\end{theorem}
\addtocounter{theorem}{-1}
}
\begin{proof}
  Every step of \Cref{distalg2} takes $O(1)$, except for step (5). So we need to show that the deterministic LLL algorithm can be run on the residual problem $\mathcal A$ in $2^{O(\sqrt{\log \log n})}$ rounds.
  
  Each $B \in \mathcal B$ survives to  $\mathcal A$ with probability $P \leq e F d^4$, and this depends only on the $c$-hop neighborhood for $c = 2$. Since $P (e \Delta)^{4 c} \leq e^9 e F d^4 d^{8} \leq 1$, we may apply Theorem~\ref{shattering-thm}.

  The residual problem can be viewed as an LLL instance with probability $q$ and dependency $d$. Since $e d q^3 \leq 1$, the randomized algorithm of \cite{pettie} can solve such LLL instances on graphs of size $N$ in $O( \log_d N)$ rounds. Here, each component of $G[R]$ has size at most $N = d^8 \log n$, so this takes time $r = O( \frac{\log d + \log \log n}{\log d} ) \leq O( \log \log n)$. Theorem~\ref{shattering-thm} therefore shows that the residual problem can be solved in $r 2^{O(\sqrt{\log \log n})} \leq 2^{\sqrt{O(\log \log n)}}$ rounds.
\end{proof}

\subsection{Examples of Events with Bounded Fragility}
We give some examples of how to compute $f(E)$ for certain types of events.

\begin{definition}[$s$-witnessable event]
Let $E$ be an event defined on variables $X(1), \dots, X(v)$. We say that $E$ is \emph{$s$-witnessable}, if for any configuration of the variables $X \in D^v$ for which $E$ is true, there exists indices $i_1, \dots, i_s$ with the following property: for any $X' \in D^v$ such that $X(i_j) = X'(i_j)$ for $j = 1, \dots, s$, the event $E$ is true on $X'$. We say that indices $i_1, \dots, i_s$ \emph{witness} that $E$ is true on $X$.
\end{definition}

As an example, consider a $\beta$-frugal coloring, i.e. a vertex-coloring with the property that each vertex has at most $\beta$ neighbors of any given color. Suppose that we assign colors to each vertex in a graph, and let $X(i)$ denote the color assigned to vertex $i$. For any vertex $i$, let $B$ be the event that the coloring assigned to $i$ fails to be frugal, i.e. some color appears at least $\beta + 1$ times in the neighborhood of $i$. This event is $\beta+1$-witnessable, although it depends on $\Delta$ variables.

\begin{proposition}
\label{thm1}
If $B$ is $s$-witnessable, then $f(B) \leq 2^s \Pr_{\Omega}(B)$.
\end{proposition}
\begin{proof}
Consider the event $E_B$ as defined in Definition~\ref{fragile-definition}. We claim that if $E_B$ holds, then there must be at least $2^{v-s}$ vectors $a' \in \{0, 1 \}^v$ such that $B$ holds on $Z_{a'}$. For, suppose that we fix $X_0, X_1$ that $B$ holds on $Z_a$. By hypothesis, there must exist indices $i_1, \dots, i_s$ such that any configuration $X'$ which agrees with $Z_a$ on coordinates $i_1, \dots, i_s$ also satisfies $B$. Thus, changing the entries of $a$ outside the coordinates $i_1, \dots, i_s$ does not falsify $B$.

Thus, for any fixed choice of $X_0, X_1$, we have
$$
[E_B] \leq \frac{\sum_{a \in \{0, 1 \}^v} [\text{$B$ holds on configuration $Z_a$}]}{2^{v-s}}
$$
where here $[E_B]$ is the indicator function for the event $E_B$, and likewise for $[\text{$B$ holds on configuration $Z_a$}]$

But note that for a fixed choice of $a$, the configuration $Z_a$ is drawn according to the distribution $\Omega$, so that overall we have $\Pr_{X_0, X_1 \sim \Omega}(E_B) \leq \frac{2^v \Pr_{\Omega}(B)}{2^{v-s}} = 2^s \Pr_{\Omega}(B)$.
\end{proof}

Another example of fragility is for a large-deviation event.
\begin{theorem}
\label{thm2}
Consider a large-deviation event $B$ defined by $\sum_{i,j} c_{ij} [X(i) = j]  \geq t$,  where $c_{ij} \in [0,1]$ and $[X(i) = j]$ is the indicator random variable. Let $\mu = \sum_{i,j} \Pr_{\Omega}(X(i) = j)$. Then for
$t \geq 2 \mu (1+\delta)$ we have
$$
f(B) \leq \Bigl( \frac{e^{\delta}}{(1+\delta)^{1+\delta}} \Bigr)^{2 \mu}
$$
\end{theorem}
\begin{proof}
Suppose we draw $X_0, X_1$ independently from $\Omega$. Then observe that $B$ is true of some configuration $Z_a$ if it is true on the configuration $Z_{a^*}$, where we define
$$
a^* = \begin{cases}
0 & \text{if $c_{i,X_0(i)} \geq c_{i,X_1(i)}$} \\
1 & \text{if $c_{i,X_1(0)} < c_{i, X_1(i)}$} \\
\end{cases}
$$

Thus, if we define $Y_i = \max( c_{i,X_0(i)}, c_{i, X_1(i)} ) $, we see that $E_B$ holds iff $\sum_i Y_i \geq t$. Each $Y_i$ is bounded in the range $[0,1]$ and $\E[Y_i] \leq \E[ c_{i,X_0(i)} + c_{i, X_1(i)}] = 2 \sum_j \Pr(X(i) = j)$, and the variables $Y_i$ are independent. The stated bound therefore follows from Chernoff's bound.
\end{proof}

\subsection{Application: $\pmb{k}$-SAT}
Consider a boolean formula $\Phi$, in which each clause contains $k$ literals, and clause overlaps with at most $d$ other clauses. A classic application of the LLL is to show that as long as $ d \leq 2^k / e$, then the formula is satisfiable; further, as shown by \cite{gst} this bound is asymptotically tight.  Using our Theorem~\ref{tr1}, we are able to show a qualitatively similar bound.

{
\renewcommand{\thetheorem}{\ref{ksat-prop}}
\begin{proposition}
If $\Phi$ has $m$ clauses and every clause intersects at most $d$ others where $d \leq e^{-10} (4/3)^{k/12} \approx 0.00005 \times 1.02426^k$, there is a distributed algorithm to find a satisfying solution to $\Phi$ is $2^{O(\sqrt{\log \log m})}$ rounds.
\end{proposition}
\addtocounter{theorem}{-1}
}
\begin{proof}
The probability space $\Omega$ is defined by selecting each variable $X(i)$ to be true or false with probability $1/2$. We have a bad-event for each clause, that it becomes violated; each bad-event $B$ has probability $p = 2^{-k}$. Thus, in total, there are $m$ bad-events.

Consider some clause $C$, wlg. $C = x_1 \vee x_2 \vee \dots \vee x_k$. Consider drawing $X_0, X_1$ independently from $\Omega$; if $X_0(i) = X_1(i) = T$ for some $i \in [k]$, then necessarily $C$ holds on every configuration $Z_a$. Thus, a necessary event for $B$ to fail on some $Z_a$ is for $X_0(i) = F$ or $X_1(i) = F$. This has probability $3/4$ for each $i$ and by independence we have $f(B) \leq (3/4)^k$.

So in our problem, $F = (3/4)^k$ (by contrast, $p = 2^{-k}$). By Theorem~\ref{tr1}, we need  $e^{10} (3/4)^k d^{12} \leq 1$, i.e. $d \leq e^{-10} (4/3)^{k/12}$.
\end{proof}

\subsection{Application: Defective Coloring}
A $h$-defective $k$-coloring of a graph $G = (V,E)$, is a mapping $\phi: V \rightarrow \{1, \dots, c \}$, with the property
that every vertex $v$ has at most $f$ neighbors $w$ with $\phi(v) = \phi(w)$. A proper vertex coloring is a $0$-defective coloring. A classical application of the iterated LLL is to show that a graph with maximum degree $\Delta$ has an $h$-defective $k$-coloring with $k = O( \Delta/f )$, for any integer $f \geq 0$. In \cite{ghaffari-lll}, a distributed algorithm was given to find such a coloring in $2^{O(\sqrt{\log \log n})}$ rounds. When $\Delta$ is small, this is a straightforward application of the distributed pLLL algorithm. For large values of $\Delta$, this required a somewhat specialized recoloring argument. We will show how to replace this recoloring argument with an application of Theorem~\ref{tr1}. 

\begin{proposition}
\label{combine-coloring}
Suppose that there is a map on the vertex set $\phi_1:V \rightarrow \{0, \dots, k_1 - 1\}$. Suppose that each induced subgraph $G[\phi^{-1} (i)]$ has an $h$-defective $k_2$-coloring, for $i = 0, \dots, k_1 - 1$. Then $G$ has an $h$-defective $k$-coloring with $k = k_1 k_2$
\end{proposition}
\begin{proof}
Let $\phi_{2,i}$ be the coloring of $G[\phi^{-1}(i)]$. Define the coloring map $\phi: V \rightarrow \{0, \dots, k-1 \}$ by $\phi(v) = k_2 i + \phi_{2,i}(v)$ for $v \in V_i$.
\end{proof}

\begin{lemma}
\label{degree-reduce}
Let $G$ have maximum degree $\Delta$. There is a distributed algorithm in $2^{O(\sqrt{\log \log n})}$ rounds to find an $h$-defective $k$-coloring of $G$ with $h =  O(\log \Delta)$ and $k = \frac{\Delta}{\log \Delta}$.
\end{lemma}
\begin{proof}
We assume throughout that $\Delta \geq \Delta_0$ for any desired constant $\Delta_0$, as otherwise the problem is trivial.

The probability space $\Omega$ assigns each vertex to a class $V_i$ with probability $1/k$. A bad-event $B_v$ is that vertex $v$ has more than $K \log \Delta$ neighbors in some class $V_i$, where $K$ is some sufficiently large universal constant. This only depends on classes of the neighbors of $v$, and so only affects another $B_w$ if $\text{dist}(v,w) \leq 2$. So the LLL dependency $d$ satisfies $d \leq \Delta^2$.

We next compute $f(B_v)$. We can write $B_v$ as $B_v = B_{v,1} \cup B_{v,2} \cup \dots \cup B_{v,k}$, where $B_{v,i}$ is the event that $v$ has too many neighbors in $V_i$. So $f(B_v) = f(B_{v,1} \cup \dots \cup B_{v,k}) \leq f(B_{v,1}) + \dots + f(B_{v,k}) \leq \Delta f(B_{v,1})$. Note that $B_{v,1}$ can be interpreted as a large-deviation event; here $\mu$ is the expected number of neighbors of $v$ entering $V_i$, which is at most $\Delta / k = \log \Delta$. Theorem~\ref{thm2} thus gives $f(B_{v,1}) \leq \Bigl( \frac{e^{\delta}}{(1+\delta)^{1+\delta}} \Bigr)^{2 \mu}$, 
where $\mu = \log \Delta$ and $\delta = K/2 - 2$. When $K = 100$, simple calculations show that $f(B_{v,1}) \leq \Delta^{-285}$ and thus $f(B_v) \leq \Delta^{-284}$.

We can therefore apply Theorem~\ref{tr1} to find a configuration avoiding all bad-events, as $e^{10} F d^{12} \leq e^{10} (\Delta \times \Delta^{-284}) \times (\Delta^2)^{12}$, which is smaller than $1$ for $\Delta$ sufficiently large.
\end{proof}

{
\renewcommand{\thetheorem}{\ref{defect-coloring-prop}}
\begin{proposition}
  Suppose $G$ has maximum degree $\Delta \geq h$. There is a distributed algorithm in $2^{O(\sqrt{\log \log n})}$ rounds to find an $h$-defective $k$-coloring with $k = O( \Delta/h )$.
\end{proposition}
\addtocounter{theorem}{-1}}
\begin{proof}
We apply Lemma~\ref{degree-reduce} to the original graph, thereby obtaining a coloring $\phi_1$ with $k_1 = \Delta/\log \Delta$ colors, such that each color class of $\phi_1$ has maximum degree $\Delta_1 = O(\log \Delta)$.

We apply Lemma~\ref{degree-reduce} again to each color class of $\phi_1$, giving a $\Delta_2$-defective $k_2'$-coloring with $k_2' = \frac{\Delta_1}{\log \Delta_1}$. By Proposition~\ref{combine-coloring}, this yields a $\Delta_2$-defective coloring $\phi_2$ with $k = k_1 k_2 = O( \frac{\Delta}{\log \log \Delta})$ colors. Note that this step can be carried out in parallel for each color class of $\phi_1$, so the overall time is $2^{O(\sqrt{\log \log n})}$.

Each color class of $\phi_2$ has maximum degree $\Delta_2 = O(\log \log \Delta)$. This degree is small enough that we can use the standard pLLL construction, using Lemma~\ref{prev-lll-thm2} and the approach of \cite{ghaffari-lll}, to find a $h$-defective coloring $\phi_3$ of each color class with $k_3' = O( \Delta_2 / h)$ colors.  This step can be carried out in parallel for each class of $\phi_2$, so overall it also takes $2^{O(\sqrt{\log \log n})}$ rounds.

Applying Proposition~\ref{combine-coloring}, this yields a final $h$-defective coloring with $k_3' k_2 = O(\Delta/h)$ colors.
\end{proof}


%% file: SLOCAL.tex
\section{Obstacles to Derandomizing Local Algorithms}
\label{sec:SLOCAL}

In this section, we discuss possible limitations to derandomizing local algorithms.

\subsection{An Exponential Separation in the \textsf{SLOCAL} Model}
\label{sec:exponentialseparation}
A key consequence of \Cref{thm:basicderandomization} is that in the \SLOCAL model, for locally checkable problems, randomized algorithms are no more powerful than deterministic algorithms up to logarithmic factors.  In this section, we show that once we start caring about logarithmic factors, an exponential separation shows up. More concretely, we show that the same problem of sinkless orientation, which was shown to exhibit an exponential separation between randomized and deterministic complexities in the \LOCAL model, exhibits an exponential separation also in the \SLOCAL model. However, the placement of the bounds are different, and in fact surprising to us. 

In the \LOCAL model, Brandt et al.~\cite{brandt} showed that randomized sinkless orientation requires $\Omega(\log\log n)$ round complexity and Chang et al.~\cite{chang16} showed that deterministic sinkless orientation requires $\Omega(\log n)$ round complexity. These were complemented by matching randomized $O(\log\log n)$ and deterministic $O(\log n)$ upper bounds by Ghaffari and Su~\cite{GS17}. In contrast, in the \SLOCAL model, the tight complexities are an exponential lower: we show that sinkless orientation has deterministic \SLOCAL complexity $\Theta(\log\log n)$ and randomized \SLOCAL complexity $\Theta(\log\log \log n)$.

\begin{theorem}\label{thm:SinklessLB} Any $\seqloc$ algorithm for sinkless orientation on $d$-regular graphs has locality $\Omega(\log_{d} \log n)$. Any $\randseqloc$ algorithm for sinkless orientation on $d$-regular graphs has locality $\Omega(\log_{d} \log \log n)$.
\end{theorem}
\begin{proof}
Suppose that there is a $\seqloc[t]$ sinkless orientation algorithm. By Proposition~\ref{prop:ZtoL}, this yields a deterministic $\loc[t d^t + t \log^*n]$ algorithm. As shown by Chang et al.\cite{chang16}, deterministic \LOCAL sinkless orientation algorithms must have round complexity $\Omega(\log n)$. So we know that $t d^t + t \log^*n \geq \Omega(\log n)$, which can easily be seen to imply that $t \geq \Omega(\log_{d} \log n)$.

The proof for the $\Omega(\log_{d} \log \log n)$ lower bound on the locality of $\randseqloc$ algorithms is similar. Here, we use the $\Omega(\log\log n)$ lower bound of Brandt et al.\cite{brandt} on the round complexity of randomized \LOCAL sinkless orientation algorithms.
\end{proof}
\begin{theorem}\label{thm:DET_SLOCAL_Sinkless} There is an  $\seqloc[O(\log\log n)]$ algorithm to compute a sinkless orientation of any graph with minimum degree at least $3$.
\end{theorem}
\begin{proof} Follows immediately from the randomized \LOCAL sinkless orientation of Ghaffari and Su~\cite{GS17}, which has round complexity $O(\log \log n)$, combined with \Cref{thm:basicderandomization}, and noticing that whether a given orientation is sinkless can be locally checked in $1$ round.
\end{proof}
\begin{theorem}\label{thm:RAND_SLOCAL_Sinkless} There is an  $\randseqloc[O(\log\log\log n)]$ algorithm to compute a sinkless orientation of any graph with minimum degree at least $3$.
\end{theorem}
\begin{proof} We describe here a clean algorithm for the special case of regular graphs with degree $d\geq 500$, which already suffices to exhibit an exponential separation in the \SLOCAL model in light of Theorem~\ref{thm:SinklessLB}. This algorithm can be extended to the general case of arbitrary graphs with minimum degree at least $3$ by adding a few small steps similar to Ghaffari and Su~\cite[Appendix A.2]{GS17}, which have \SLOCAL complexity $O(1)$.

Below, we describe the algorithm as a two-pass  $\randseqloc[O(\log\log\log n)]$ algorithm. This algorithm is borrowed from Ghaffari and Su~\cite{GS17} almost verbatim, modulo the small change that the second pass/phase uses a deterministic \SLOCAL algorithm. As shown in Ghaffari, Kuhn, and Maus~\cite{ghaffari2017complexity}, this can then be transformed into a single-pass $\randseqloc[O(\log \log \log n)]$ algorithm.
\begin{algorithm}[H]
  \caption{An $\randseqloc$ algorithm for sinkless orientation }\label{alg:random}
  \begin{algorithmic}
\small
\State \textbf{Pass 1:}
\State Mark each edge with probability $\frac{1}{4}$.
\State For each marked edge, orient it randomly with probability $1/2$ for each direction.
\State For each node $v$, mark $v$ as {\it a bad node} of the following types according to these rules:
\begin{itemize}
	\item Type I. If $v$ has more than $d/2$ marked edges incident to it.
	\item Type II. If $v$ is not Type I but it has at least one neighbor of Type I.
	\item Type III. If $v$ is not Type I or Type II but it has no outgoing marked edges.
\end{itemize}
\State Unmark all the edges incident to Type I nodes.
\State For an unmarked edge $e$ in which both endpoints are good nodes, arbitrarily orient $e$.
\State For an unmarked edge $e$ with exactly one good endpoint, treat $e$ as a half-edge attached only to its bad endpoint.
\State
\State \textbf{Pass 2:}
\State Run the deterministic \SLOCAL algorithm of \Cref{thm:DET_SLOCAL_Sinkless} on the induced subgraph of the bad nodes (including half-edges).
\end{algorithmic}
\end{algorithm}
The first pass clearly has locality $O(1)$. By the analysis of Ghaffari and Su~\cite{GS17}, each connected component on bad nodes has size at most $N=O(\log n)$ and therefore, by \Cref{thm:DET_SLOCAL_Sinkless}, the second pass has locality $O(\log\log N) = O(\log\log \log n)$. 
\end{proof}

\subsection{Impossibility of Derandomizing Non-Locally-Checkable Problems}
\label{sec:SLOCAL2}
In this section, we show  that the local checkability condition is an important property in derandomization, and it is not there just for technicalities. Specifically, let us consider the following simple toy problem: Given a cycle of size $n$, where $n$ is known to all nodes, we should mark $(1\pm o(1)) \sqrt{n}$ nodes. The trivial zero-round randomized algorithm, which marks each node with probability $1/\sqrt{n}$, succeeds with high probability. We next argue that any deterministic algorithm needs $\Omega(\sqrt{n})$ rounds.
\begin{proposition}\label{prop:cyclemarking}
 Any deterministic $\SLOCAL$ algorithm for the cycle-marking problem needs $\Omega(\sqrt{n})$ rounds.
\end{proposition}
\begin{proof}
  Let $A$ be a deterministic $\SLOCAL$ algorithm with locality $T=\sqrt{n}$. Consider some $50\sqrt{n}$ disjoint cycles each of length $n$, where the $i^{\text{th}}$ cycle $C_i$ is made of nodes with IDs $(i-1)n+1$ to $in$. Assuming for simplicity that $2 T$ is an integer dividing $n$, we can break the cycle up into $n/T$ strips of length $T$. We then run $A$ using a permutation $\pi$ which processes all the odd strips before all the even strips. Within each strip, the permutation can be arbitrary. That is, $A$ first processes vertices $(i-1) n + 1, \dots, (i-1) n + T,  (i-1) n + 2 T + 1, \dots, (i-1) n + 3 T, $ and so on, and then goes back to process $(i-1) n + T+1, \dots, (i-1) n + 2 T$, etc.  Observe that, due to the spacing of the strips, the output of any given vertex $v$ in strip $j$ depends only on the information in strips $j-2, \dots, j+2$.

The algorithm $A$ must mark at least one vertex per cycle $C_i$ (in fact, it must mark $(1-o(1)) \sqrt{n}$ of them). For each cycle $i$, let $u_i$ be some vertex marked by $A$.

Now consider the graph $G'$ where we carve out the $5$ strips around each vertex $u_i$, and concatenate these all into a single cycle. We can run algorithm $A$ on $G'$ with the same permutation $\pi$ (skipping vertices which are no longer present in $G'$). Each vertex $u_i$ remains marked since it its neighboring strips are still present in $G'$.   There may be some additional marked edges as well, but overall $G'$ has at least $50 \sqrt{n}$ marked nodes. The size of $G'$ is $50 \sqrt{n} \times 5 \times T = 250 n$, so this is too many marked nodes and hence $A$ fails on $G$.
  \end{proof}

\subsection{Complete Problems}
\label{subsec:complete}

We next prove completeness results for several natural and widely-studied distributed graph problems. To relate different problems to each other, we use the notion of \emph{locality-preserving reductions} as defined in \cite{ghaffari2017complexity}. A distributed graph problem $\mathcal{A}$ is called \emph{polylog-reducible} to a distributed graph problem $\mathcal{B}$ iff a $\polylog n$-time deterministic \LOCAL algorithm for $\mathcal{A}$ (for all possible $n$-node graphs) implies a $\polylog n$-time deterministic \LOCAL algorithm for $\mathcal{B}$. In addition, a distributed graph problem $\mathcal{A}$ is called $\polyseqloc$-complete if a) problem $\mathcal{A}$ is in the class $\polyseqloc$ and b) every other distributed graph problem in $\polyseqloc$ is polylog-reducible to $\mathcal{A}$. As a consequence, if any \polyseqloc-complete problem could be solved in $\polylog n$ deterministic time in the \LOCAL model, then we would have $\polyloc = \polyseqloc$ and thus every problem in $\polyseqloc$ (and thus by \Cref{prop:ZtoS} also all problems in $\polyzloc$) could also be solved in $\polylog n$ time deterministically in the \LOCAL model.  The following completeness proofs can therefore be understood as conditional lower bounds: efficiently derandomizing any of the following problems would also efficiently derandomize the whole class $\polyzloc$ of polylog-time distributed Las Vegas algorithms.

In the \emph{minimum set cover} problem, we are given a ground set $X$ and a collection $\mathcal{S}\subseteq 2^X$ of subsets of $X$ which covers $X$. The objective is to find a smallest possible collection of sets $\mathcal{C}\subseteq \mathcal{S}$ such that $\bigcup_{A\in \mathcal C} A = X$.  In the \emph{distributed set cover problem} and often also more generally in distributed linear programming algorithms (cf.\ \cite{bartal97,nearsighted,papa93}), the set system $(X,\mathcal{S})$ is modeled as a bipartite graph with a node for every element $x\in X$ and every set $A\in \mathcal{S}$. There is an edge between $x\in X$ and $A\in\mathcal{S}$ if and only if $x\in A$.  \begin{theorem}\label{thm:setcover_complete}
  Approximating distributed set cover by any factor $\alpha = \text{polylog} n$ is \polyseqloc-complete. 
\end{theorem}
\begin{proof}
  We first note that \cite{ghaffari2017complexity} has shown that the problem of computing a
  $(1+\eps)$-factor approximation for distributed set cover
  is in the class $\polyseqloc $ for any
  $\eps \geq 1/\polylog n$. To show that the problem is also $\polyseqloc$-hard, we reduce from
  the problem of computing a $\polylog(n)$-color conflict-free
  multicoloring (abbreviated CFM) of an $n$-node hypergraph $H = (V,E)$ with $|E| \leq \text{poly}(n)$.  This problem was shown to
  be \polyseqloc-complete in Theorem 1.10 of
  \cite{ghaffari2017complexity} (even for almost uniform hypergraphs).
  A $q$-color CFM of $H$ assigns a
  non-empty set of colors from $\set{1,\dots,q}$ to each vertex, 
  such that for each hyperedge $f\in E$, there is one color that is
  assigned to exactly one node in $f$ \cite{smorodinsky2013conflict}.

  Assume that there is a deterministic polylogarithmic-time algorithm
  for computing a $\alpha$-approximate solution for a given
  distributed set cover problem. We will use this to recursively construct a polylogarithmic-time deterministic CFM algorithm for $H$ with $q=\poly(\alpha\cdot\log n)$
  colors.

  Let $E_-$ be the set of hyperedges of $H$ of size at most
  $\delta := C \alpha(\ln n + 1)$, for some constant $C$, and let $E_+ := E \setminus E_-$. Lemma 6.2 of
  \cite{ghaffari2017complexity} shows for hypergraphs of
  $\polylog n$ rank, a CFM with $\polylog n$
  colors can be computed deterministically in $\polylog n$ time in the
  \LOCAL model. As $\alpha \leq \text{polylog}(n)$,  we can therefore compute
  a $\polylog n$-color CFM of $H_- := (V,E_-)$
  in polylogarithmic deterministic time. 

It remains to also compute a
  CFM of the hypergraph $H_+ := (V,E_+)$. Since $H_+$ has minimum rank $\delta$, ia standard probabilistic argument using independent rounding and alteration shows that $H_+$ has a vertex cover of size $O(\frac{|V| \log n }{\delta})$.  Our approximation algorithm for set cover therefore gives a vertex cover $U$ of size $|U| \leq O( \frac{|V| \alpha \log n }{\delta} )$.

  Let $F = \{ f \cap U \mid f \in E_+ \}$. Apply the algorithm recursively to obtain a $q'$-color CFM of the hypergraph $(U,F)$; then extend it to a $q'+1$-color CFM of $H_+$ by assigning one additional color to every node in $V - U$.

  If $C$ is a sufficiently large constant, then $|U| \leq |V|/2$. Therefore, this process terminates after a polylogarithmic number of steps. Furthermore, in each iteration of this recursion, we add $\polylog n$ colors --- $1$ color to extend the CFM of $(U,F)$ to $H$ and $\polylog n$ colors for $H_{-}$. Since there are $\polylog n$ steps, the total number of colors used is $\polylog n$ and the runtime is also $\polylog n$.
  \end{proof}

Given a graph $G=(V,E)$, the minimum dominating set (MDS) problem asks for a
smallest possible vertex set $D\subseteq V$ such that for all
$u\in V$, either $u\in D$ or some neighbor $v$ of $u$ is in $D$. The MDS problem is essentially equivalent to the minimum set cover problem, and the following theorem shows that also approximating MDS is \polyseqloc-complete. MDS is one of the most widely studied problems in
the \LOCAL model. Polylogarithmic-time randomized distributed
approximation algorithms have been known for a long time, e.g.,
\cite{dubhashi05,jia02,dominatingset,nearsighted}.

\begin{theorem}\label{thm:domset_complete}
  Approximating MDS by a polylogarithmic
  factor is \polyseqloc-complete.
\end{theorem}
\begin{proof}
  Since the dominating set problem is a special case of the set cover
  problem, the problem is clearly also in the class \polyseqloc. To
  show that the problem is \polyseqloc-hard, we reduce from the
  distributed set cover problem. Assume that we are given a
  distributed set cover problem $(X,\mathcal{S})$. Recall the communication graph of the problem is given by the
  bipartite graph $G$ with nodes $X\cup \mathcal{S}$ and an edge between $x\in X$
  and $A\in \mathcal{S}$ if and only if $x\in A$. We assume that every $x \in X$ is contained in at least one set $A \in \mathcal S$, as otherwise the set cover instance has no solution.
  
  Let us define a graph $G'$, which has the same nodes of $G$ and all the edges of $G$.  In addition, $G'$ has an edge between
  any two nodes $A$ and $B$ in $\mathcal{S}$ if $A \cap B \neq \emptyset$. A communication
  round on the graph $G'$ can be simulated in $O(1)$ round on graph
  $G$.

 We claim  that any dominating set $D$ of $G'$ can
  directly be converted to a set cover of size at most $|D|$ for
  $(X,\mathcal{S})$, and vice-versa; this will immediately show that $\alpha$-approximation algorithms for MDS lead to $\alpha$-approximation algorithms for set-cover.
  
By construction of $G'$, any set cover $C\subseteq \mathcal{S}$ is already a dominating set of $G'$.
Further, let $D$ be a dominating set of $G'$. If $D$ contains
  nodes in $X$, we can convert $D$ into a dominating set $D'$ of size
  $|D'|\leq |D|$ such that $D'$ only contains nodes from
  $\mathcal{S}$. For each $x\in D\cap X$, we replace $x$ by some
  neighboring node $A\in \mathcal{S}$ (recall that such a neighbor
  always exists). Since $A$ is connected to all $B\in \mathcal{S}$ for
  which $x\in B$, $A$ covers all nodes of $G'$ that were covered by
  $x$ and thus $D'$ is still a dominating set and because $D'\subseteq
  \mathcal{S}$, it is also directly corresponds to a solution of the
  set cover instance $(X,\mathcal{S})$.
\end{proof}

We next consider a natural greedy packing problem similar to MIS, one of the four classic symmetry breaking problems.
(We do not know whether the MIS problem is
\polyseqloc-complete.) Consider a bipartite
graph $G=(L \cup R, E)$, with $n = |L \cup R|$. We call $L$ the left side and $R$ the right
side of $G$. For a positive integer $k$, we define a $k$-star of $G$
to be a left-side node $u\in L$ together with $k$ right-side neighbors
$v_1\dots,v_k\in R$. A \emph{maximal independent
  $k$-star set} is a maximal set of pairwise vertex-disjoint $k$-stars
of $G$.\footnote{The problem of finding a maximal independent
$k$-star set is equivalent to the problem of finding an MIS of the graph $H$, whose nodes correspond to $k$-stars of $G$.
However, if $G$ has $n$ nodes, the graph $H$ can have up to
$O(n^{k+1})$ nodes and thus a polylogarithmic-time MIS algorithm on
$H$ does not directly lead to a polylogarithmic-time algorithm for
finding a maximal independent $k$-star set on $G$.}

\begin{theorem}\label{thm:maximalstars_complete}
  Let $G=(L\cup R,E)$ be an $n$-node bipartite graph, where every node
  in $L$ has degree at most $\Delta$. For every
  $\lambda\geq 1/\polylog n$, the problem of finding a maximal
  independent $\lceil\lambda\Delta\rceil$-star set of $G$ is
  \polyseqloc-complete. 
\end{theorem}
\begin{proof}
When processing a node $u\in L$, we can decide whether one can
  add a $k$-star with $u$ as the center by inspecting the $2$-hop
  neighborhood of $u$. The problem  is in the class
  \seqloc(2), and so is clearly in \polyseqloc.

  To show \polyseqloc-hardness, we reduce from the minimum 
  set cover problem. Consider a minimum set cover problem
  $(X,\mathcal{S})$ and the corresponding bipartite graph $G=(L\cup R,
  E)$, where $L$ corresponds to the sets $\mathcal{S}$, the $R$
  corresponds to the set of elements $X$. Assume that the maximum
  degree of the nodes in $L$ is at most $\Delta$. All of the sets in set cover instance have size at most
  $\Delta$, and so the minimum set cover has size at least 
  $|X|/\Delta$. 

Assume that we are given a maximal independent
  $\lceil\lambda\Delta \rceil$-star set of $G$. Because each star
  contains $\lceil\lambda\Delta \rceil$ nodes from $X$, such a set can
  consist of at most $|X|/(\lambda\Delta)$ stars. Further, by adding
  all the sets corresponding to the centers of the stars to the set
  cover, by the maximality of the star set, the maximum set size of
  the remaining set cover instance is at most $(1-\lambda)\Delta$.
  Repeating $O\big(\frac{\log\Delta}{\lambda}\big)$ times therefore yields a
  set cover solution that is optimal up to a factor
  $O\big(\frac{\log\Delta}{\lambda^2}\big)$.
\end{proof}

The last problem we consider in this section is \emph{sparse neighborhood cover} of a graph \cite{Awerbuch-Peleg1990}, which is a fundamental structure with a large number of applications in distributed systems
\cite{peleg00}. 

\begin{definition}[Sparse neighborhood cover, adapted from
  \cite{Awerbuch-Peleg1990}]\label{def:sparsecover}
  Let $G=(V,E)$ be a graph and let $r\geq 1$ be an integer parameter.
  A sparse $r$-neighborhood cover of $G$ is a collection of clusters
  $C_1,\dots,C_k$ such that each cluster has diameter at most
  $r\cdot \polylog n$, such that each $r$-hop neighborhood of $G$
  is completely contained in at least one of the clusters, and such
  that every node of $G$ is contained in at most $\delta \leq \polylog n$ of the
  clusters.  
\end{definition}

It is shown in \cite{Awerbuch-Peleg1990} that such neighborhood covers exist (even
when replacing all the $\polylog n$ terms in the definition by terms
of order $\log n$).

\begin{theorem}
  For any $r \leq \polylog n$, the problem of computing a sparse
  neighborhood cover of a graph is \polyseqloc-complete.
\end{theorem}
\begin{proof}
It suffices to consider $r=1$, as a sparse $1$-neighborhood cover for $G^r$ yields a sparse $r$-neighborhood cover for $G$. For any vertex $v$, define $N^{+}(v)$ to be the inclusive one-hop neighborhood of $v$, i.e. $\{v \} \cup N(v)$.

  The sequential construction of \cite{Awerbuch-Peleg1990} to
  construct sparse neighborhood covers almost directly gives a multi-phase \SLOCAL algorithm. The clusters are constructed
  iteratively in $O(\log n)$ passes by a ball growing argument starting from the center node
  of the cluster. To construct a cluster around a node $u$, only the
  $O(r\log n)$-hop neighborhood of $u$ has to be inspected. The
  construction can thus be turned into a concatenation of $O(\log n)$
  \SLOCAL($O(r\log n)$)-algorithms. Lemma 2.3 of
  \cite{ghaffari2017complexity} then implies that the construction can
  be turned into a single \SLOCAL($O(r\log^2 n)$)-algorithm. For
  $r\leq\polylog n$, the problem of computing a sparse neighborhood
  cover is therefore in \polyseqloc.

  We show that the problem is \polyseqloc-hard by reducing from MDS. Assume that we want to compute a
  dominating set of a graph $G$. We first compute a sparse
  $1$-neighborhood cover $C_1, \dots, C_k$ of $G$.
  
  For each cluster $C_i$, let $C_{i,\text{inside}}$ be
  the set of nodes in $C_i$ that are not at the boundary (i.e. the nodes $v$ with $N^{+}(v) \subseteq C_i$). Note that $C_i$ dominates $C_{i, \text{inside}}$. Since $C_i$ has diameter $\polylog(n)$, we can compute a set $U_i \subseteq C_i$ such that $U_i$ dominates $C_{i, \text{inside}}$ and $U_i$ has the smallest cardinality of any such dominating set.\footnote{If we want the local computations at the nodes to
    be polynomial, it is sufficient to just to select $U_i$ in a greedy fashion;  this will lose an additional $\log n$ factor in the approximation ratio.}  We then define our dominating set $D = U_1 \cup \dots \cup U_k$; this is a valid dominating  set because for each vertex $v$ we have $N^{+}(v) \subseteq C_i$ for some value $i$, and hence $v \in C_{i, \text{inside}}$ must be dominated by $U_i$.

  It remains to show that $D$ is a good
  approximation. Let $D^*$ be an optimal dominating set. Define a corresponding set $U_i' = C_i \cap D^*$. We claim that $U_i'$ dominates $C_{i, \text{inside}}$. For, any $v \in C_{i, \text{inside}}$ is dominated by some $w \in N^{+}(v) \cap D^*$, and we will have $w \in C_i$ as well, so $v$ is dominated by $U_i'$.
  
  Since $U_i$ has the smallest cardinality of any such dominating sets, we have $|U_i| \leq |U_i'|$. So 
  $$
  |D| = \sum_i |U_i| \leq \sum_i |U_i'| = \sum_{v \in D^*} (\text{\# clusters $C_i$ containing $v$})
  $$
  
  By definition, every vertex is in at most $\delta$ clusters, so this quantity is at most $\delta |D^*|$. The theorem now
  follows because $\delta\leq \polylog n$ and because by
  \Cref{thm:domset_complete}, computing a $\polylog$-approximate
  dominating set is \polyseqloc-complete.
\end{proof}

It was clear before that given a polylog-time deterministic distributed algorithm
for network decomposition, we can also get such an algorithm for
computing sparse neighborhood covers. The above theorem shows that
also the converse is true: Given a polylog-time deterministic
distributed algorithm for sparse neighborhood covers, we can get a
polylog-time deterministic distributed algorithm to compute a
$(O(\log n),O(\log n))$-network decomposition. Hence, up to
polylogarithmic factors, the complexity of the two key graph
clustering variants are equivalent in the \LOCAL model.


%% file: shattering.tex
\section{Residual problems and shattering}
\label{shattering-appendix}
This section provides an overview of the shattering method for distributed algorithms. This method was first developed in \cite{barenboim2016locality}, and variants have been applied to many graph algorithms since then such as \cite{gmis, barenboim2016locality, elkin20152delta}. We aim here to provide a comprehensive and unified treatment of these types of algorithms. Also, while our arguments will closely parallel those of \cite{barenboim2016locality}, the latter does not give us precisely the parameters needed for our LLL algorithms.
 
The shattering method has two phases. In the first phase, there is some random process. This random phase satisfies most of the vertices $v \in G$, and we will fix the choices these ``good'' vertices make.  In the second phase, we let $R \subseteq V$ denote the unsatisfied vertices. These vertices are very sparse, and the connected components of $G[R]$ are relatively small. We then use a deterministic algorithm to solve the residual problem on (each component of) $G[R]$.  We assume that each vertex $v$ survives to $R$ with probability $p$, and this bound holds even for adversarial choices for the random bits outside the $c$-hop neighborhood of $v$. We assume here that $c \geq 1$ is some constant, and all asymptotic notations in this section may hide dependence upon $c$. We will assume throughout this section that $p$ satisfies
$$
p \leq (e \Delta)^{-4 c}
$$

\begin{definition}
Given a graph $G = (V,E)$ and a vertex set $W \subseteq V$, we say that $W$ is \emph{connected in $G$} if for every $w, w' \in W$ there is a path in $G$ from $w$ to $w'$.  We say that $G$ is \emph{connected} if $V$ is connected in $G$.
\end{definition}

\begin{definition}
Given a graph $G = (V,E)$ and a parameter $c \geq 1$, we say that a set $U \subseteq V$ is a \emph{$c$-backbone of $G$} if $U$ is an independent set of $G^c$, and $G^{3 c}[U]$ is connected.

We say that a set $W \subseteq V$ is \emph{$c, m$-backbone-free in $G$} if there is no $c$-backbone $U \subseteq W$ with $|U| \geq m$. 
\end{definition}

\begin{proposition}
\label{count-prop}
If $G$ has maximum degree $\Delta$ and $x$ is a vertex of $G$, then there are at most $(e \Delta)^{3cm}$ distinct $c$-backbones $U$ such that $|U| = m, x \in U$.
\end{proposition}
\begin{proof}
Define $r = \Delta^{3 c}$. For any $v \in V$, define 
$$
\tau_i(v) = \sum_{\substack{\text{$c$-backbones $U$} \\  |U| \leq i \\ U \ni v}} (e r)^{-|U|}
$$

We also define $\tau_{\infty}(v) = \lim_{i \rightarrow \infty} \tau_i(v)$. Thus
\begin{align*}
\sum_{\substack{\text{$c$-backbones $U$} \\ U \ni x \\ |U| = m}} 1 &= (e r)^m \sum_{\substack{\text{$c$-backbones $U$} \\ U \ni x \\ |U| = m}} (e r)^{-|U|} \leq (e r)^m \sum_{\substack{\text{$c$-backbones $U$} \\ U \ni x}} (e r)^{-|U|} = (e r)^m \tau_{\infty}(x)
\end{align*}

We show by induction on $i$ that $\tau_i(v) \leq 1/r$ for all integers $i \geq 0$ and $v \in V$. Consider some $c$-backbone $U$ for $x$ of size at most $i$. Let $v_1, \dots, v_k$ denote the elements of $U$ within distance $3 c$ of $x$; we must have $k \leq r$. For each $j = 1, \dots, k$, let $U_j$ denote the vertices $u \in U$ which are reachable from $v_j$ but not $x, v_1, \dots, v_{j-1}$ via paths in $G^{2 c}$. Clearly $U_1, \dots, U_k$ are $c$-backbones containing $v_1, \dots, v_k$ respectively, and have size strictly less than $i$. Also note $(e r)^{|U|} = (e r)^{-1} (e r)^{-|U_1|} \dots (e r)^{-|U_k|}$.

Summing over all choices of possible $U_1, \dots, U_k$ and using the induction hypothesis, we have:
\begin{align*}
\sum_{\substack{\text{$c$-backbones $U$} \\  |U| \leq i \\ U \ni v \\ |U \cap N_2(v)| = \{v_1, \dots, v_k \}}} (e r)^{-|U|} &\leq ( e r )^{-1} \tau_{i-1}(v_1) \dots \tau_{i-1}(v_k) \leq ( e r )^{-1} r^{-k}
\end{align*}

Summing over all possible choices for $k$ and $v_1, \dots, v_k$ gives:
$$
\tau_i(v) \leq\sum_{k=0}^{r} \binom{r}{k} \frac{1}{e r^{k+1}} = \frac{ (1 + 1/r)^r}{e r} \leq 1/r
$$

This implies that $\tau_{\infty}(v) \leq 1/r$ for all $v \in V$, and the claim follows.
\end{proof}

\begin{proposition}
\label{backbone-exists-prop}
Suppose that $G = (V,E)$ and $S \subseteq V$. Suppose that there is some $W \subseteq S$ such that $W$ is an independent set of $G^{c}$ and $W$ is connected in $G^{c}[S]$. Then $G$ has a $c$-backbone $U \subseteq S$ with $|U| \geq |W|$.
\end{proposition}
\begin{proof}
Let $S' \subseteq S$ denote the set of vertices $s \in S$ such that $s$ is reachable from $W$ via paths in $G^c[S]$. Since $W$ is connected in $G^c[S]$, the graph $G^c[S']$ is connected.

Let $U$ be a maximal independent set of $G^c[S']$. We claim that $|U| \geq |W|$. Note that for every $u \in U$, there is at most one $w \in W$ with $d(u,w) \leq c$ (since $W$ is independent in $G^{c}$). So if $|U| < |W|$, there must exist some $w \in W$ which has $d(u,w) > c$ for all $u \in U$. Since $W \subseteq S'$, this would contradicting maximality of $W$.

Next we claim that $G^{3 c}[U]$ is connected. For, consider any pair $u, u' \in U$. Since $U \subseteq S'$ and $G^c[S']$ is connected, there is a path $w = x_1, x_2, \dots, x_k = w'$ with $x_1, \dots, x_k \in S'$ and $\text{dist}_G(x_i, x_{i+1}) \leq c$. By maximality of $U$, for each $i = 1, \dots, k$ there is some $v_i \in U$ with $\text{dist}_G(v_i, x_i) \leq c$. Now note that $w, v_1, v_2, \dots, v_k, w'$ is a path in $G^{3 c}[U]$.
\end{proof}

\begin{proposition}
\label{backbone-free-prop}
W.h.p., $R$ is $c,m$-backbone-free for $m = \Omega(\log n)$.
\end{proposition}
\begin{proof}
First, observe that if there is a $c,m'$-backbone $U'$ in $R$ of size $m' > m$, then one can remove nodes as needed to form a $c,m$-backbone $U$. Thus, it suffices to show that there are no $c,m$-backbones in $R$.

Consider $x \in V$. There are at most $(e \Delta)^{3 c m }$ possible backbones of size $m$ including $x$. Since a backbone is an independent of $G^{3 c}$, each such backbone survives to $R$ with probability $p^m$. Summing over all $x$, we see that the expected number of surviving $c,m$-backbones is at most $n (e \Delta)^{3 c m} p^m \leq n (e \Delta)^{-c m}$;  this is smaller than $n^{-100}$ for $m \geq \Omega(\frac{\log n}{\log \Delta})$. 
\end{proof}

\begin{theorem}
  \label{con-size0}
  Let $v \in V$. Then the probability that the component of $v$ in $G[R]$ has size exceeding $w$, is at most
  $( e \Delta)^{-w/(\Delta+1)^c + 1}$.
\end{theorem}
\begin{proof}
  If $v \notin R$, then the component of $v$ in $G[R]$ is the empty set and the bound holds trivially. Otherwise, suppose that $v \in R$. Let $T$ be the connected component of $G[R]$ containing $v$ and let $W$ be chosen to be maximal such that $v \in W$ and $W \subseteq T$ and $W$ is an independent set of $G^c$. Since $G^c$ has maximum degree $\Delta^c$, we have $|W| \geq |T|/(\Delta+1)^c$. 

We also claim that $G^{3c}[W]$ is connected. For, consider any vertices $x, x' \in T$. There must exist a path $x = y_1, y_2, \dots, y_k = x'$ in $G[R]$. By maximality of $U$, for each $i = 1, \dots, k$ there is some $v_i \in U$ with $\text{dist}_G(v_i, x_i) \leq c$ (otherwise one could add $x_i$ to $U$). Thus $x, v_1, v_2, \dots, v_k, x'$ is a path in $G^{3 c}[W]$.

So, $W$ is a $c, m$-backbone with $m \geq |T|/(\Delta+1)^c$ and $W \ni v$. If $|T| \geq w$ then by removing vertices we can construct a $c,m'$ backbone $W'$ with $m' = \lceil \frac{w}{(\Delta+1)^c} \rceil$, such that $W'$ survives in $G[R]$ and $v \in W'$.

There are at most $(e \Delta)^{3 c m'}$ possible backbones, and each survives in $R$ with probability $p^{m'}$. Thus, the overall probability that some such backbone exists is at most $( (e \Delta)^{3 c} p )^{w/(\Delta+1)^c + 1} \leq (e \Delta)^{-\frac{w}{(\Delta+1)^c} + 1}$.
\end{proof}

\begin{theorem}
\label{con-size2}
W.h.p., one can obtain an $(O(\log \log n), O( (\log \log n)^2 )$-network-decomposition of $G[R]^r$ in $r 2^{O(\sqrt{\log \log n})}$ rounds, for any integer $r \geq 1$.
\end{theorem}
\begin{proof}
We begin by finding a $(2, O(\log \log n))$-ruling set $W$ for the graph $J = G^{c}[R]$; this step can be performed in $O(\log \log n)$ rounds using the algorithm of \cite{schneider2013symmetry}. This ensures that every $v \in R$ has distance $\text{dist}_J(v, W) \leq t = \Theta(\log \log n)$.

For each $v \in R$, let $s(v)$ be the vertex $w \in W$ which minimizes $\text{dist}_G(v,w)$ (breaking ties arbitrarily). Since $W$ is a ruling set we have $\text{dist}_J(v,s(v)) \leq t$, and this implies $\text{dist}_G(v,s(v)) \leq c t$.  

Now consider the graph $H$, on vertex set $W$, and with an edge $(w,w')$ if there are $v, v' \in R$ with $s(v) = w, s(v') = w'$ and $\text{dist}_{G[R]}(v, v') \leq r$.

We claim that every connected component of $H$ must have size at most $m = \log n$.  For, suppose that $T \subseteq W$ is connected in $H$. We claim that $T$ is connected in $J$. Since $W$ is a connected component of $H$, it suffices to show that for any edge of $H$ between vertices $w, w' \in T$, there is a path from $w$ to $w'$ in $J$. By definition, there are $v, v' \in R$ with $s(v) = w, s(v') = w'$ and $v \sim v'$. So $\text{dist}_J(v, w) \leq t < \infty$ and $\text{dist}_J(v', w') \leq t < \infty$. Also, $\text{dist}_{G[R]}(v, v') \leq r < \infty$ and thus $\text{dist}_J(v,v') < \infty$ as well.

In addition, since $W$ is a ruling set of $J$, it must be that $T$ is an independent set of $G^c$. Applying Proposition~\ref{backbone-exists-prop} with $S = R$, we see that $G$ has a $c$-backbone $U \subseteq R$ of size $|U| \geq |T|$; by Proposition~\ref{backbone-free-prop} this implies w.h.p. that $|T| \leq O(\log n)$.

Since every connected component of $H$ has size $O(\log n)$, we can use the deterministic algorithm of \cite{panconesi95} in $2^{O(\sqrt{\log \log n})}$ rounds to obtain an $(\lambda,D)$-network-decomposition of $H$ with $\lambda, D \leq O(\log \log n)$. These are rounds with respect to communication on the graph $H$, each of which can be simulated in $O(r t)$ rounds on the graph $G$. Thus, this step requires $r 2^{O(\sqrt{\log \log n})}$ communications rounds in $G$ in total.

Let $X_1, \dots, X_{\lambda}$ denote the color classes of this decomposition and let $L = G[R]^r$; we now generate a network-decomposition $X'_1, \dots, X'_{\lambda}$ of the graph $L$, specifically we set $X'_j = s^{-1}(X_j).$

In order to show this decomposition works, consider some $X'_j$ and consider the induced subgraph $K = L[X'_j]$. Now observe that for any vertex $x \in K$, we have $\text{dist}_K(x, s(x)) \leq c t$. For by definition, $s(x)$ is the closest vertex of $W$ to $x$ in $G[R]$. By definition of $s$, every vertex $y$ along the shortest path from $x$ to $s(x)$ will also have $s(y) = s(x)$. Since $\text{dist}_G(x, s(x)) \leq c t$, this implies $\text{dist}_K(x, s(x)) \leq c t$ as well.

Let $v, v' \in K$ with $\text{dist}_K(v, v') < \infty$; we need to show that $\text{dist}_K(v,v') \leq O( ( \log \log n)^2 )$. 

There is a path $v = a_1, a_2, \dots, a_{\ell} = v'$ in $K$. For each $i = 1, \dots, \ell$ let $w_i = s(a_i)$ where each $w_i \in X_j$. Each $(w_i, w_{i+1})$ is an edge of $H[X_j]$ and so $w_1, w_{\ell}$ are connected in $H[X_j]$.

Since $H[X_j]$ has diameter $O(\log \log n)$, there must exist a path $w_1 = u_1, \dots, u_{k} = w_{\ell}$ in $H[X_j]$ with $k \leq O(\log \log n)$. Also, since each $(u_i, u_{i+1})$ is an edge in $H$, there must be $b_1, \dots, b_{k} \in R$ such that $u_i = s(b_i)$ and $\text{dist}_J(b_i, b_{i+1}) \leq r$. Since $s(b_i) \in X_j$, we see that $b_1, \dots, b_k \in X'_j$.

Since $\text{dist}_{G[R]}(b_i, b_{i+1}) \leq r$, we have $b_i \sim b_{i+1}$ in the graph $L$, which implies that $\text{dist}_K(b_i, b_{i+1}) \leq 1$.

We may thus compute $\text{dist}_{K}(u_i, u_{i+1}) \leq \text{dist}_K(u_i, b_i) + \text{dist}_K(b_i, b_{i+1}) + \text{dist}_K(b_{i+1}, u_{i+1}) \leq 2 c t + 1$. This implies that $\text{dist}_K(w_1, w_{\ell}) \leq k (2 c t + 1)$, and therefore $\text{dist}_K(a_1, a_{\ell}) \leq \text{dist}_K(a_1, w_1) + \text{dist}_K(w_1, w_{\ell}) + \text{dist}_K(w_{\ell}, a_{\ell}) \leq 2 c t + k (2 c t + 1) \leq O( (\log \log n)^2 )$.
\end{proof}

{
\renewcommand{\thetheorem}{\ref{shattering-thm}}
\begin{theorem}
   Suppose each vertex survives to a residual graph $R$ with
  probability at most $(e \Delta)^{-4 c}$, and this bound holds even
  for adversarial choices of the random bits outside the $c$-hop
  neighborhood of $v$ for some constant $c \geq 1$.

  Then w.h.p each connected component of the residual graph has size at most $O(\Delta^{2 c} \log n)$.

  If the residual problem can be solved via a $\seqloc[r]$
  procedure $A$, then the residual problem can be solved w.h.p.\ in the \LOCAL
  model in $\min \big \{ r 2^{O(\sqrt{\log \log n})}, O( r(\Delta(G^r) + \log^* n) ) \big \}$ rounds.
\end{theorem}
\addtocounter{theorem}{-1}
}
\begin{proof}
  
 Let us first show the bound on the components of $G[R]$. Let $w = a \Delta^{2 c} \log n$ for some constant $a > 0$ to be determined. Taking a union bound over $v \in V$ and applying Theorem~\ref{con-size0}, we see that the probability that there is some $v \in V$ with component size exceeding $w$ is at most
  $$
  n \times (e \Delta)^{-w/(\Delta+1)^c + 1} \leq n \times (e \Delta)^{-\Omega(a \Delta \log n)} \leq 1/n
  $$
  for $a$ sufficiently large.
  
For the first runtime bound, by Proposition~\ref{prop:ZtoL} the algorithm $A$ leads to a deterministic algorithm in  $O(r \Delta(G^{2 r}) + r \log^*n )$ rounds, which can be run on each component of the residual graph.

For the second runtime bound, use Theorem~\ref{con-size2} to get a $(D,C)$-network decomposition of $G[R]^r$ in $2^{O(\sqrt{\log \log n})}$ rounds, with $C = O(\log \log n)$ and $D = O( (\log \log n)^2 )$. By Proposition~\ref{prop:StoL}, this allows us to run $A$ in $O( C (D+1) r) = r 2^{O(\sqrt{\log \log n})}$ rounds deterministically.
\end{proof}

Proposition~\ref{prop:ZtoL} gives us a rich source of SLOCAL algorithms. In particular, by combining Proposition~\ref{prop:ZtoL} with Theorem~\ref{shattering-thm}, we can use Las Vegas algorithms for the second phase of shattering algorithms.